\documentclass[10pt,journal]{IEEEtran}

\makeatletter
\def\ps@headings{%
\def\@oddhead{\mbox{}\scriptsize\rightmark \hfil \thepage}%
\def\@evenhead{\scriptsize\thepage \hfil \leftmark\mbox{}}%
\def\@oddfoot{}%
\def\@evenfoot{}}

\makeatother

\usepackage{graphicx}
\usepackage[noadjust]{cite}
\usepackage{tabularx,ragged2e,booktabs,caption}
\usepackage{times}
\usepackage{alltt}
\usepackage{verbatim}
\usepackage{moreverb}
\usepackage{amsmath}
\usepackage{amssymb}
\usepackage{url}
\usepackage{epsfig}
\usepackage{epstopdf}
\usepackage{subfigure}
\usepackage{multirow}
\usepackage[ruled,vlined,linesnumbered]{algorithm2e}
\usepackage{mathrsfs}
\usepackage{amsthm}
\usepackage{listings}
\usepackage[none]{hyphenat}
\usepackage{xcolor}
\usepackage{cases}

\usepackage{caption}
\usepackage{graphicx,subfigure}
\usepackage{float}
\usepackage[caption = false]{subfig}

\usepackage{tabularx}
\usepackage{adjustbox}

\newtheorem{theorem}{Theorem}
\newtheorem{proposition}{Proposition}
\newtheorem{definition}{Definition}

\newtheorem{conjecture}{Conjecture}

\begin{document}

\title{Queueing Delay Minimization in Overloaded Networks}

\author{Xinyu~Wu,~\IEEEmembership{Student Member,~IEEE,}
       Dan~Wu,~\IEEEmembership{Member,~IEEE,}
       Eytan~Modiano,~\IEEEmembership{Fellow,~IEEE}
       \thanks{Xinyu Wu and Eytan Modiano are with the Laboratory for Information and Decision Systems, Massachusetts Institute of Technology, Cambridge, MA, 02139 USA. email: xinyuwu1@mit.edu, modiano@mit.edu.}
       \thanks{Dan Wu is with School of Electrical and Electronic Engineering, Huazhong University of Science and Technology, Wuhan, China. email: danwuhust@hust.edu.cn.}
       \thanks{}
       \thanks{Under journal review.}
       }


\maketitle

\IEEEpeerreviewmaketitle

\begin{abstract}
We develop link rate control policies to minimize the queueing delay of packets in overloaded networks. We show that increasing link rates does not guarantee delay reduction during overload. We consider a fluid queueing model that facilitates explicit characterization of the queueing delay of packets, and establish explicit conditions on link rates that can minimize the average and maximum queueing delay in both single-hop and multi-stage (switching) networks. These min-delay conditions require maintaining an identical ratio between the ingress and egress rates of different nodes at the same layer of the network. We term the policies that follow these conditions \emph{rate-proportional} policies. We further generalize the rate-proportional policies to \emph{queue-proportional} policies, which minimize the queueing delay asymptotically based on the time-varying queue length while remaining agnostic of packet arrival rates. We validate that the proposed policies lead to minimum queueing delay under various network topologies and settings, compared with benchmarks including the backpressure policy that  maximizes network throughput and the max-link-rate policy that fully utilizes bandwidth. We further remark that the explicit min-delay policy design in multi-stage networks facilitates co-optimization with other metrics, such as minimizing total bandwidth, balancing link utilization and node buffer usage. This demonstrates the wider utility of our main results in data center network optimization in practice.

\end{abstract}

\section{Introduction}

Network overload occurs when the total demand of network users exceeds the network capacity \cite{le2010optimal,georgiadis2006optimal}. Severe overload leads to heavy congestion with high delay and packet loss so that network performance is seriously degraded \cite{zhang2022aequitas, shah2011fluid, venkataramanan2013queue}. Network overload occurs more frequently in data center networks, notably driven by the upsurge of machine learning applications and higher traffic demands due to the larger size of training models  \cite{openai2018ai, zheng2023traffic, narayanan2021efficient}. Meanwhile, the slowdown of Moore's law compared with traffic growth further increases the likelihood of overload \cite{ballani2018bridging,cai2021understanding}. These facts pose challenges for network service providers to utilize communication bandwidth effectively to maintain network performance under overload \cite{poutievski2022jupiter, mellette2017rotornet, ballani2020sirius}. Another source of overload is capacity reduction due to network breakdowns, including network drains during data center maintenance  \cite{singh2015jupiter, poutievski2022jupiter}, unexpected failures of nodes and links \cite{al2008scalable, poutievski2022jupiter}, and cyberattacks such as denial-of-service \cite{fu2019fundamental, fu2019network} and node hijacking \cite{sermpezis2018survey, al2016bgp}. {Previous research investigated network overload extensively in various aspects: throughput maximization \cite{tekin2012dynamic,katsalis2013dynamic}, overload balancing \cite{li2014dynamic,georgiadis2006optimal,li2014receiver,chan2010fairness,qu2017mitigating}, traffic dynamics and their convergence \cite{talreja2008fluid,perry2013fluid,perry2016chattering}, and rapid recovery from overload \cite{perry2015achieving}.}


In this paper, we focus on \emph{delay minimization} in overloaded networks. Reducing network delay is crucial to both network users and enterprises, given delay-sensitive applications such as short-form videos and quantitative trading \cite{wu2022queueing}, and   service level delay objectives are difficult to meet under overload in large-scale data centers \cite{zhang2022aequitas}. Enterprise revenues are sensitive to delay: Google reports that advertisement revenues will decrease by $20\%$ if web search delay increases from 0.4s to 0.9s, and Amazon reports that an extra 100ms response time decreases the sales by $1\%$ \cite{sun2016delay}. 

The delay increase caused by network overload is mainly due to the \emph{queueing} delay of packets, since the queue buffers are increasingly backlogged due to the overload. A common approach to reducing queueing delay is active queue management, which drops packets when overload is observed or the queueing delay exceeds a specific upper bound \cite{athuraliya2001rem,cho2020overload,addanki2022abm}. Alternative approaches include network calculus-based scheduling with worst-case latency guarantees \cite{cruz1991calculus,zhang2022aequitas}, traffic shaping and pacing \cite{mellanox,Swift2020}, and smart buffer design to absorb traffic spikes \cite{broadcom}. However, these approaches are heuristic with no performance guarantees. Understanding optimal policies that globally minimize queueing delay holds potential for further delay reduction, which remains a hard problem \cite{ji2012delay,neely2012delay,bertsekas2021data}. {Extensive previous work demonstrates the power of load balancing in achieving close-to-zero queueing delay over heavily-loaded parallel servers \cite{sun2016delay, hsieh2017delay, liu2018achieving, wang2019delay, weng2020achieving}. Techniques like Kleinrock Independence Approximation have been applied to approximate the mean queueing delay of packets in general networks \cite{kleinrock2007communication, pang1986approximate,takagi1988queuing,modiano1996simple}. However, developing policies to minimize the queueing delay remains intractable.}

In this paper, we design optimal policies for queueing delay minimization in overloaded networks. Network overload further raises technical challenges in characterizing queueing delay using stochastic models since Little's law \cite{bertsekas2021data}, the foundation of queueing delay analysis in stationary systems, no longer holds in overloaded networks as the long-term expectation of delay is infinite. Hence, policies that achieve close-to-minimum delay when the network is not overloaded may no longer perform well in overloaded networks. We give an intuitive example in the $2\times 1$ single-hop network of Fig. \ref{fig:2x1-single-hop}, where the link capacities are $c_1=4$, $c_2=2$, and external packets arrive at node $s_i$ with rate $\lambda_i~ (i=1,2)$ and are transmitted to the shared buffer at node $d$ whose service rate is $\mu = 2$. Suppose that at most one link can be activated at a time. In this case, the maxweight scheduling policy that activates the link that is connected to the source node with longer queue backlog \cite{hsieh2017delay} has been shown to have near-optimal delay performance when the network is not overloaded $(\lambda_1+\lambda_2<\mu)$. However, the maxweight scheduling policy fails in delay minimization for packets injected into node $s_2$ during overload when $(\lambda_1, \lambda_2)=(8,3)$, as $s_1$ always has longer queue backlog than $s_2$, thus blocking packets in $s_2$ until the overload ends. For the case where the simultaneous activation of the two links is allowed, we show later in this paper that neither the throughout-optimal backpressure policy \cite{tassiulas1990stability} nor serving packets with maximum link rates can minimize delay under overload, while instead fixing the rate of link $(s_1,d)$ to be $2$, and link $(s_2,d)$ to be $0.75$, can minimize the delay in this example. These counter-intuitive observations reveal the necessity to redesign link rate control policies for queueing delay minimization in overloaded networks.

\begin{figure}[!htbp]
\centering
\includegraphics[width=0.82\linewidth]{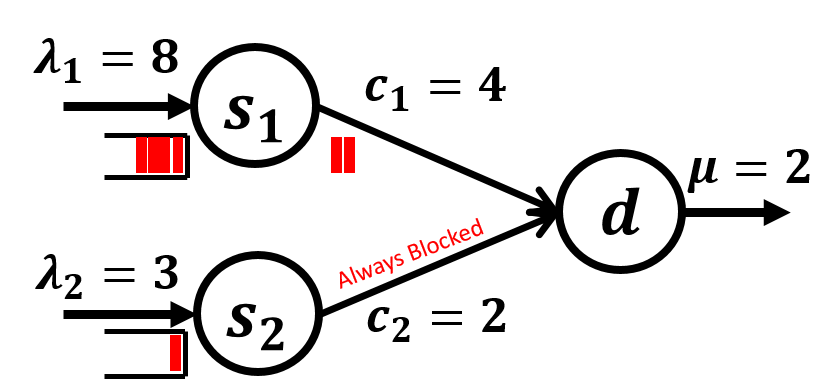}
\caption{An example of a $2\times 1$ single-hop network: Under maxweight policy, link $(s_1,d)$ is always activated while $(s_2,d)$ is blocked since the queue in $s_1$ grows with rate $\lambda_1-c_1=4$ while the queue in $s_2$ grows with rate at most $\lambda_2=3$.}
\label{fig:2x1-single-hop}
\end{figure}


To this end, we develop a deterministic fluid queueing model that elegantly solves the technical challenges of queueing delay characterization in overloaded networks. 
The fluid model regards network traffic as continuous flows instead of discrete packets. It well approximates the discrete packet transmission when the time unit is sufficiently small. 
Based on the model, we develop link rate control policies that minimize the queueing delay of the packets that arrive to the network within a bounded time interval, which corresponds to the duration of the overload. We demonstrate that our proposed policies minimize queueing delay in both single-hop and multi-stage switching networks, which serve as the basic structure of data center networks including Clos \cite{zhao2019minimal, zhang2021gemini, singh2015jupiter} and Tree \cite{heller2010elastictree, leiserson1985fat}. Hence, our results shed light on policy design to reduce delay in data centers under overload. 

We summarize our contributions as follows. (i) 
We derive explicit conditions on link rates that minimize both the average and maximum queueing delay of packets in general single-hop and multi-stage networks. 
These conditions correspond to a \emph{rate-proportional} policy which maintains an identical ratio between the ingress and egress rates of each node at the same layer, i.e., the ingress rates of all the nodes at a layer should be proportional to their egress rates. (ii) We generalize the rate-proportional policies to \emph{queue-proportional} policies, that can minimize queueing delay asymptotically based on real-time queue backlogs, and do not require knowledge of packet arrival rates \cite{wu2022overload, weng2020achieving}. (iii) We validate that our proposed policies achieve minimum delay in various settings of single-hop and multi-stage networks, and demonstrate further delay reduction compared with benchmarks including the backpressure policy \cite{georgiadis2006optimal, tassiulas1990stability} that maximizes network throughput and the max-link-rate policy that fully utilizes bandwidth. (iv) We demonstrate that the proposed explicit min-delay policies  facilitate co-optimization with other metrics that are important in data centers: minimizing total bandwidth, balancing link utilization, and balancing overload rates at different node buffers. We also determine a set of more relaxed min-delay conditions for  tree data center structures, and conjecture on the sufficient and necessary condition minimizing queueing delay in general multi-stage networks.




The rest of the paper is organized as follows: Section \ref{sec:model} introduces the models, definitions, and problem formulation for queueing delay minimization in overloaded networks; Section \ref{sec:static_delay_optimal} investigates the rate-proportional policies that minimize the queueing delay in both single-hop and multi-stage networks; Section \ref{sec:queue_policy} generalizes the rate-proportional policies to queue-proportional policies; Section \ref{sec:evaluation} validates the min-delay performance of the proposed policies and the delay reduction compared with the benchmarks; and Section \ref{sec:extension} discusses the practical and theoretical extension of the proposed min-delay policies. 

\section{Models, Definitions and Problem Formulation}
\label{sec:model}

In this section, we introduce the single-hop and multi-stage network models and the fluid queueing model of packet flows, and define network overload. We then characterize the queueing delay of a packet based on the fluid model and formulate the queueing delay minimization problem based on the derived explicit forms of the average and maximum queueing delay of packets.

\subsection{Network Models: Topology and Dynamics}



\subsubsection{Single-hop networks} A single-hop network contains a set of ingress nodes and egress nodes. We model an $N_S\times N_D$ single-hop network as a bipartite graph $(\mathcal{V},\mathcal{E})$ with $\mathcal{V}:=\{\mathcal{V}_S,\mathcal{V}_D\}$, where $\mathcal{V}_S$ denotes the set of ingress nodes with size $|\mathcal{V}_S| = N_S$, and $\mathcal{V}_D$ denotes the set of egress nodes with size $|\mathcal{V}_D| = N_D$, and $\mathcal{E}$ denotes the set of transmission links from $\mathcal{V}_S$ to $\mathcal{V}_D$. Fig.~\ref{fig:single-hop}(a) visualizes the single-hop structure. Examples of single-hop networks include switched networks as Fig.~\ref{fig:single-hop}(b) and server farms as Fig.~\ref{fig:single-hop}(c). The single-hop structure is the basic network unit that constitutes many data center networks \cite{al2008scalable,singh2015jupiter}.


\begin{figure}[!htbp]
\centering
\includegraphics[width=0.98\linewidth]{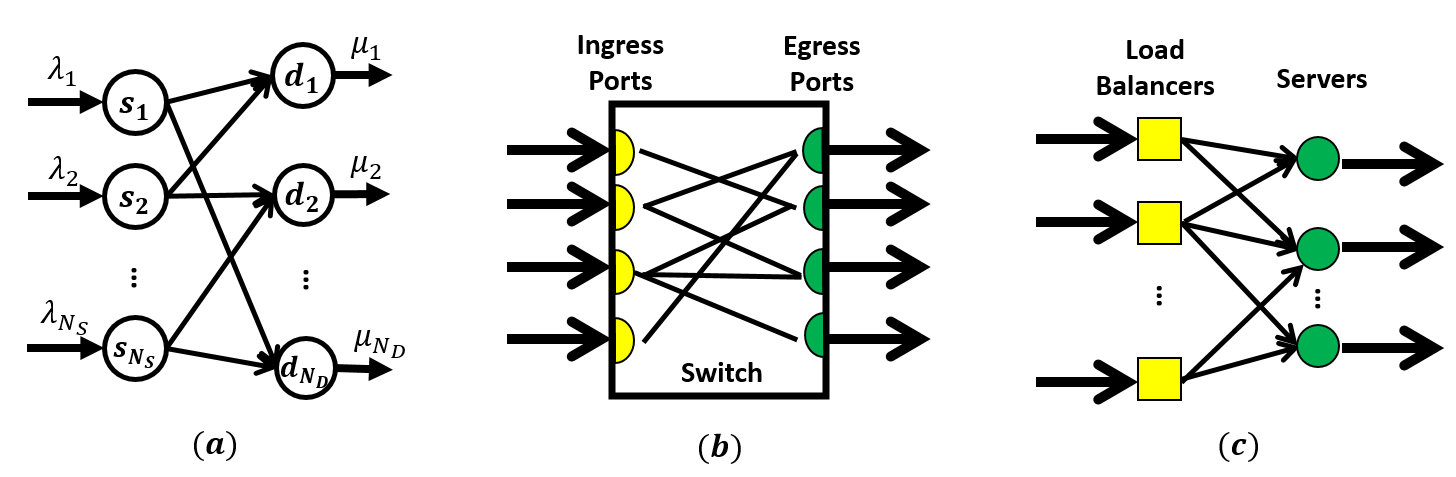}
\caption{(a) A single-hop network structure; (b) A switched network with ingress and egress ports; (c) A server farm with load balancers as ingress and servers as egress.}
\label{fig:single-hop}
\end{figure}


We denote the $i$th ingress node by $s_i$ and the $j$th egress node by $d_j$ in single-hop networks. {We consider that a packet injected into an ingress node can be dispatched to any connected egress node and depart.} We denote the packet arrival rate at ingress node $s_i$ by $\lambda_i$, which represents the average number of packets injected into node $s_i$ in a time unit. 
We use $\boldsymbol{\lambda}:=\{\lambda_i\}_{i=1}^N$ to represent the packet arrival rate vector {which we assume to be static (i.e., time-invariant)}. We further assume that at each node, packets in the buffer follow the first-come-first-serve service which is common in real network infrastructures \cite{bertsekas2021data}. We denote the queue length in node $k$ at time $t$ by $q_k(t)$. 
We denote the packet transmission rate on link $(s_i,d_j)$ at time $t$ by $g_{s_id_j}(t)$, which represents the number of packets transmitted over $(s_i,d_j)$ at time $t$. Each link $(s_i,d_j)$ is associated with a capacity value $c_{s_id_j}$, which is the maximum transmission rate, i.e., $0\leq g_{s_id_j}(t) \leq c_{s_id_j},~\forall t,~\forall (s_i,d_j)\in \mathcal{E}$. Note that trivially $g_{s_id_j}(t)\equiv 0$ for any $(s_i,d_j)\notin \mathcal{E}$, and $g_{s_id_j}(t)=0$ when $q_{s_i}(t)=0$ for any $(s_i,d_j)\in \mathcal{E}$, which means no packets will be transmitted through $(s_i,d_j)$ when there is no queue backlog in node $s_i$.  We use $\mathbf{g}(t):=\{g_{s_id_j}(t)\}_{(s_i,d_j)\in \mathcal{E}}$ to denote the transmission rate vector and $\mathbf{c}:=\{c_{s_id_j}\}_{(s_i,d_j)\in \mathcal{E}}$ to denote the capacity vector. We consider that each egress node $d_j$ serves packets in a work-conserving manner with its maximum service rate denoted by $\mu_{j}$, whenever there exists queue backlog in the buffer. It is clear that work-conserving service at the egress nodes is a necessary condition for queueing delay minimization. Therefore we can merely focus on setting the transmission rate vector $\mathbf{g}(t)$ between the ingress and egress nodes to minimize queueing delay.

We apply a fluid queueing model to characterize the queueing dynamics: Packets are modeled as continuous traffic flows instead of discrete units, which means the queue length can be fractional. The fluid model is based on the flow conservation law, which states that the net increase of queue length equals to the difference between the number of new arrivals and departures at a node at any time, i.e.,
\begin{equation}
\label{eqn:ODE}
\begin{cases}
\dot{q}_{s_i}(t)=\lambda_i-\sum_{d_j\in \mathcal{V}_D} g_{s_id_j}(t), ~\forall i = 1,\dots,N_S \\
\dot{q}_{d_j}(t)=\sum_{s_i\in \mathcal{V}_S} g_{s_id_j}(t)-g_{d_j}(t), ~\forall j = 1,\dots,N_D
\end{cases}
\end{equation}
where under the work-conserving mechanism at egress nodes,  $g_{d_j}(t) := \mu_j$ if $q_{d_j}(t)>0$.
The dynamics \eqref{eqn:ODE} 
provide a simplified framework for flow control analysis compared with the discrete queueing model \cite{georgiadis2006optimal}. Note that it is different from the fluid model defined in some prior works which captures the scaled limit of the queue backlog \cite{dai2005maximum,shah2011fluid,markakis2018delay}, an indicator for queue stability but not suited to study queueing delay. 

\subsubsection{Multi-stage networks} We extend the definitions to multi-stage networks. A multi-stage network contains multiple layers of network nodes, and transmission links connecting nodes at adjacent layers. A multi-stage network with $L$ layers of nodes is composed of an ingress layer where packets are injected into the network, an egress layer where packets depart from the network, and $L-2$ middle layers between them. We index the layers in order where the ingress layer is layer $1$ and the egress layer is layer $L$. We can view an $L$-layer network as a cascade of $L-1$ single-hop networks, where a packet at any node at layer $l$ can be dispatched to any of its connected node at layer $l+1$, and finally departs at some node at the egress layer. Each packet will traverse one node in each layer, and all the traversed nodes form the path of this packet. Different packets may take different paths. Fig.~\ref{fig:multi-stage} gives an example of a multi-stage network with $L=4$. Multi-stage networks are the common structures in data center infrastructures like Fat-tree \cite{al2008scalable, lebiednik2016survey}, Clos \cite{heller2010elastictree, zhao2019minimal}, and the direct-connect topology with spine blocks removed \cite{poutievski2022jupiter, wang2023topoopt}.

We use the following notations in multi-stage networks. Denote the set of nodes at layer $l$ by $\mathcal{V}_l$ with size $|\mathcal{V}_l| = N_l$, the $i$-th node at layer $l$ by $n_i^l$, and the transmission rate and capacity of link $(n_i^l, n_j^{l+1})$ between layer $l$ and $l+1$ by $g_{n_i^l,n_j^{l+1}}$ and $c_{n_i^l,n_j^{l+1}}$ respectively. The packet arrival rate to the ingress node $n_i^1$ is $\lambda_i$, and the maximum service rate at the egress node $n_j^L$ is $\mu_j$. Similarly, all the egress nodes operate in a work-conserving manner, and $g_{n_{i}^{l-1}, n_j^{l}}(t) = 0$ if $q_{n_{i}^{l-1}}(t) = 0$ and $g_{n_i^{L}}(t) = 0$ if $q_{n_i^{L}}(t) = 0$. We define the queueing dynamics in multi-stage networks in \eqref{eqn:ODE_multi_stage}, which is an extension of \eqref{eqn:ODE} from 2 layers to $L$ layers.  

\begin{equation}
\label{eqn:ODE_multi_stage}
\begin{cases}
\dot{q}_{n_i^{1}}(t)=\lambda_i-\sum\limits_{n_j^{2} \in \mathcal{V}_2} g_{n_i^{1}, n_j^{2}}(t), ~\forall i = 1,\dots,N_1 \\
\quad \\
\dot{q}_{n_i^{l}}(t)=\sum\limits_{n_{k}^{l-1} \in \mathcal{V}_{l-1}} g_{n_k^{l-1},n_i^{l}}(t) - \sum\limits_{n_{j}^{l+1}\in \mathcal{V}_{l+1}} g_{n_i^{l}, n_{j}^{l+1}}(t), 
\\ \qquad \qquad \qquad \qquad \qquad \forall i = 1,\dots,N_l, ~\forall l = 2, \dots, L-1 \\
\quad \\
\dot{q}_{n_i^{L}}(t)=\sum\limits_{n_{k}^{L-1}\in \mathcal{V}_{L-1}} g_{n_{k}^{L-1}, n_i^{L}}(t)-g_{n_i^{L}}(t), ~\forall i = 1,\dots,N_L
\end{cases}
\end{equation}


\begin{figure}[!htbp]
\centering
\includegraphics[width=0.98\linewidth]{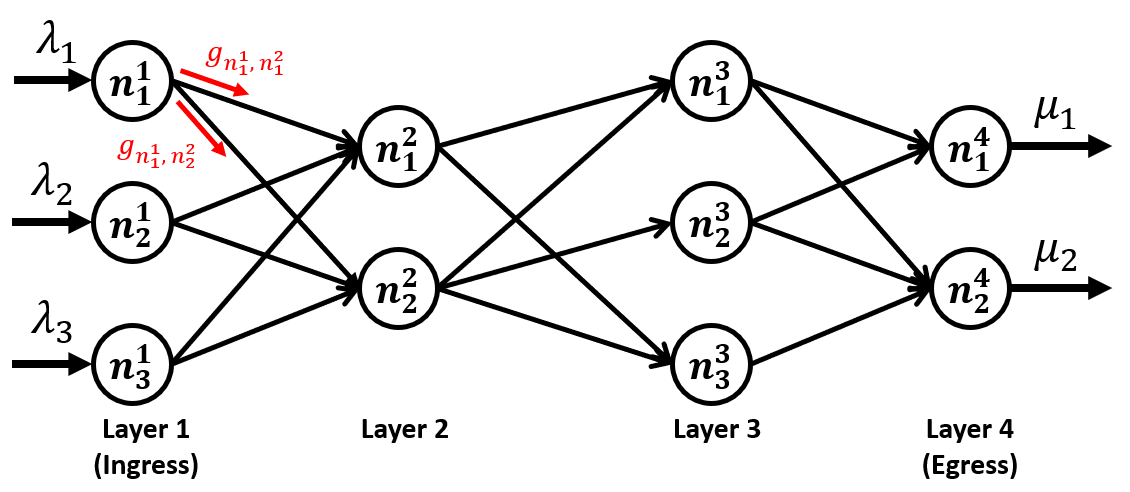}
\caption{An example of a 4-layer multi-stage network}
\label{fig:multi-stage}
\end{figure}

In this work, {we start from the special case of queueing dynamics  \eqref{eqn:ODE} and \eqref{eqn:ODE_multi_stage} under \emph{static} transmission policies where the transmission rates ${g}_{ij}(t)$ of each link $(i,j)$ at different times $t$ are a constant value $g_{ij}$ when $q_i(t)>0$.} We demonstrate below that the study over static policies facilitates the characterization of queueing delay of packets and the policy design for delay minimization, and the results inspire the \emph{dynamic} policy design based on real-time queue backlog information instead of packet arrival rates at the ingress layer. 





\textbf{Remark}: The equations in \eqref{eqn:ODE} and \eqref{eqn:ODE_multi_stage} give the general formulation  of queueing dynamics without restrictions on packet routing, where the packets at a node can be transmitted to any of its connected nodes at the next layer. We can add routing constraints for example forcing $g_{ij}(t) \equiv 0$ which means packets at node $i$ cannot be dispatched to $j$. We derive the policies that minimize queueing delay under the unrestricted dynamics and demonstrate that they still hold with routing restrictions, explained in Section \ref{subsec:practical}.

\subsection{Network Overload}

We say that a network is overloaded if there is no transmission policy that guarantees bounded queueing backlog over all the node buffers in the network. 
\begin{definition}
\label{def:overload_general}
A network is overloaded if there is no transmission policy $\{\mathbf{g}(t)\}_{t\geq 0 }$ that can guarantee $\lim_{t\rightarrow \infty}q_i(t) < \infty,~\forall i \in \mathcal{V}$. 
\end{definition}
Definition \ref{def:overload_general} requires that no transmission policy can stabilize the network, which can be interpreted as the packet arrival rate vector $\boldsymbol{\lambda}$ beyond the network capacity region \cite{neely2010stability}. We can purely focus on overloaded networks since if $\boldsymbol{\lambda}$ is interior to the capacity region, there must exist a static policy $\mathbf{g}$ which guarantees that the total egress link rates of any node is greater than its total ingress traffic rate. Then it is trivial to apply this policy so that the queueing delay is zero under the deterministic fluid queueing model.

Under static transmission policies, we can derive more explicit conditions for a single-hop and a multi-stage network being overloaded in Definition \ref{def:overload_single_hop} and \ref{def:overload_multi_stage} respectively.

\begin{definition}
\label{def:overload_single_hop}
An $N_S \times N_D$ single-hop network under static policies is overloaded if there is no transmission rate vector $\mathbf{g}$ such that $g_{ij} \in [{0}, {c}_{ij}],~\forall (i,j)\in \mathcal{E}$ and
\begin{equation}
\label{eqn:overload_single_hop}
\begin{cases}
\sum_{d_j: (s_i,d_j) \in \mathcal{E}} g_{s_id_j} \geq \lambda_i ,~\forall i = 1,\dots,N_S \\
\sum_{s_i: (s_i,d_j) \in \mathcal{E}} g_{s_id_j}  \leq \mu_j,~\forall j = 1,\dots,N_D \\
\end{cases}
\end{equation}
\end{definition}


\begin{definition}
\label{def:overload_multi_stage}
An $L$-layer network under static policies is  overloaded if there is no transmission rate vector $\mathbf{g}$ such that $g_{ij} \in [{0}, {c}_{ij}],~\forall (i,j)\in \mathcal{E}$ and
\begin{equation}
\label{eqn:overload_multi_stage}
\begin{cases}
\sum_{n_j^{2} \in \mathcal{V}_2} g_{n_i^{1},n_j^{2}} \geq \lambda_i ,~\forall i = 1,\dots,N_1 \\
\sum_{n_k^{l-1}\in \mathcal{V}_{l-1}} g_{n_k^{l-1},n_i^{l}} \leq \sum_{n_j^{l+1} \in \mathcal{V}_{l+1}} g_{n_i^{l},n_j^{l+1}}, \\ 
\qquad \qquad \qquad \forall i = 1,\dots,N_l, ~\forall l = 2, \dots, L-1  \\
\sum_{n_k^{L-1} \in \mathcal{V}_{L-1}} g_{n_k^{L-1},n_i^{L}}  \leq \mu_i,~\forall i = 1,\dots,N_L
\end{cases}
\end{equation}
\end{definition}







\subsection{Queueing Delay Characterization}

We characterize the queueing delay of packets in overloaded networks under static transmission policies. The queueing delay dominates other network delays including preprocessing delay, transmission delay, and propagation delay in overloaded network, since the overload leads to severe increase of queue backlog in node buffers. Moreover, the other network delays are  independent of the packet arrival rates and transmission policies. Therefore, we ignore the other delays in our analysis, and the delay only represents the queueing delay below. 

We derive the explicit form of the total queueing delay of a packet that arrives at an ingress node. We first consider a 2-node network with a single link in Fig.~\ref{fig:two-node-model} to explain the derivation. Consider the shaded packet at the tail of node 1. We assume it arrives at node 1 at time $t$. The queueing delay of this packet at node 1 is $q_1(t)/g_{12}$, as the shaded packet has to wait for all of the packets ahead of it to be served. The packet departs from node 1 and arrives at node 2 at time $t^{\prime}:=t+q_1(t)/g_{12}$, and thus its queueing delay at node 2 is $q_2\left(t^{\prime}\right)/\mu$. Therefore the total queueing delay for this packet is

\begin{equation}
\label{eqn:two-node-model}
\small
\begin{aligned}
\frac{q_1(t)}{g_{12}} + \frac{q_2\left(t^{\prime}\right)}{\mu} &= \frac{q_1(t)}{g_{12}} + \max\left\{\frac{q_2(t)+ \frac{q_{1}(t)}{g_{12}}(g_{12}-\mu)}{\mu},0\right\}
\\&=\max\left\{\frac{q_1(t)+q_2(t)}{\mu}, \frac{q_1(t)}{g_{12}}\right\}
\end{aligned}
\end{equation}
where the queue growth rate at node 2 is $g_{12} - \mu$, and thus $q_2\left(t^{\prime}\right)$ is equal to $q_2(t)$ plus the total growth of packets in the buffer over time length $q_1(t)/g_{12}$. The $\max$ term in the second line is to take into account that $q_2(t^{\prime})$ may reach $0$ when $g_{12}<\mu$. 


\begin{figure}[!htbp]
\centering
\includegraphics[width=0.88\linewidth]{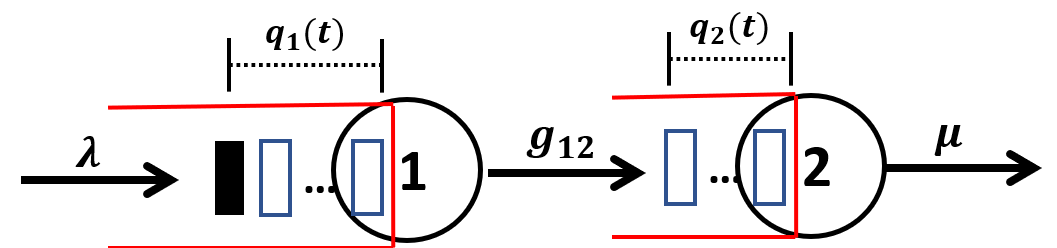}
\caption{An example of queueing delay characterization of a packet passing two nodes}
\label{fig:two-node-model}
\end{figure}

We extend the derivation in \eqref{eqn:two-node-model} to an $N_S\times N_D$ single-hop network. We denote by $D_{s_id_j}(t)$ the queueing delay of a packet that arrives at the ingress node $s_i$ at time $t$ and departs at the egress node $d_j$. The queueing delay of this packet at $s_i$ is ${q_{s_i}(t)} /{\sum_{d_k: (s_i,d_k)\in \mathcal{E}} g_{s_id_k}}$ where $\sum_{d_k: (s_i,d_k)\in \mathcal{E}} g_{s_id_k}$ is the sum of the egress link rates over all the links starting from $s_i$. Suppose that this packet is dispatched to $d_j$. The time it arrives at $d_j$ is $t^{\prime} = t +  q_{s_i}(t)/{\sum_{d_k: (s_i,d_k)\in \mathcal{E}} g_{s_id_k}}$, and its queueing delay at $d_j$ is $q_{d_j}(t^{\prime}) / \mu_j$. We consider the case where the static transmission rate vector $\mathbf{g}$ guarantees $q_{d_j}(t) > 0,~\forall t$, i.e., node $d_j$ keeps serving with rate $\mu_j$. In this case, we can express the total delay of this packet as
\begin{equation}
\small
\label{eqn:two-hop-delay}
\begin{aligned}
&D_{s_id_j}(t) = \frac{q_{s_i}(t)}{\sum_{d_k} g_{s_id_k}}  + \frac{q_{d_j}(t^{\prime})}{\mu_j}
\\& = \frac{q_{s_i}(t)}{\sum_{d_j} g_{s_id_j}} + \frac{1}{\mu_j}\left(q_{d_j}(t) + \frac{q_{s_i}(t)}{\sum_{d_k} g_{s_id_k}} \left( \sum\nolimits_{s_k} g_{s_kd_j} - \mu_j \right) \right)
\\& = \frac{1}{\mu_j}\left(q_{d_j}(t) + \frac{\sum_{s_k} g_{s_kd_j}}{\sum_{d_k} g_{s_id_k}} q_{s_i}(t)\right)
\end{aligned}
\end{equation}
where the queueing delay of a single packet that arrives at the network at any time $t$ can be expressed by a linear combination of the queue length at time $t$ at the ingress and egress nodes that this packet traverses. {We can generalize \eqref{eqn:two-hop-delay} to multi-stage networks with $L$ layers to characterize the total queueing delay of packets taking any path $p$, denoted by $D_p(t)$, as a linear combination of the queue length at time $t$ at all the nodes on the path.} We show in later sections that the explicit forms like \eqref{eqn:two-node-model} and \eqref{eqn:two-hop-delay} facilitate the derivation of static policies $\mathbf{g}$ that minimize queueing delay.





\subsection{Problem Formulation}


We define two queueing delay metrics that we try to minimize in this work: 
(i) the average delay $\bar{D}_{\text{avg}}$ (ii) the maximum ingress delay $\bar{D}_{\text{max}}$, whose  definitions are introduced later. At a high level, $\bar{D}_{\text{avg}}$ reflects the overall delay performance of all the arrived packets, which in practice is relevant to data centers where the overall performance is important \cite{zhang2022aequitas,poutievski2022jupiter,zhang2021gemini}; $\bar{D}_{\text{max}}$ represents the largest delay of the packets that arrive at different ingress nodes, which in practice is related to the fairness and flow completion time of tasks parallelized to different ingress nodes \cite{zats2012detail,chowdhury2015coflow,wang2023topoopt}. We focus on minimizing both metrics for packets that arrive to the network in some \emph{bounded} time interval $[t_0, t_0+T]$ where $t_0$ is the initial timestamp and $T<\infty$ is the time duration, given that network overload is a temporary event in practice. 

We give the formal definitions of $\bar{D}_{\text{avg}}$ and $\bar{D}_{\text{max}}$. We first consider an $N_S\times N_D$ single-hop network. Denote the average queueing delay of packets that arrive at the ingress node $s_i$ within $[t_0, t_0+T]$ by $\bar{D}_i$, which is
\begin{equation}
\label{eqn:bar_D_i}
\bar{D}_i = \frac{1}{T}\int_{t_0}^{t_0+T} \sum_{j=1}^{N_D} \left(\frac{g_{s_id_j}}{\sum_{k=1}^{N_D} g_{s_id_k}} D_{s_id_j}(t)\right) dt
\end{equation} 
for $\forall i=1,\dots,N_S$, where $D_{s_id_j}(t)$ is as in \eqref{eqn:two-hop-delay}. Note that \eqref{eqn:bar_D_i} contains two layers of averaging: (i) averaging over different arrival times $t$ within $[t_0, t_0+T]$, which is an unweighted integral; (ii) averaging over packets sent to different egress nodes, which is weighted by $g_{s_id_j}/{\sum_{k=1}^{N_D} g_{s_id_k}}$, i.e., the portion of packets that arrive at $s_i$ and will depart from $d_j$. With \eqref{eqn:bar_D_i}, we formulate the two delay metrics to be optimized in this paper as

\begin{subnumcases}
\\
\bar{D}_{\text{avg}} = \sum_{i=1}^{N_S} \frac{\lambda_i}{\sum_{j=1}^{N_S} \lambda_j} \bar{D}_i \label{eqn:metric_avg};
\\
\bar{D}_{\text{max}} = \max_{i=1,\dots,N_S} \bar{D}_i \label{eqn:metric_max}.
\end{subnumcases}

The $\bar{D}_{\text{avg}}$ in \eqref{eqn:metric_avg} introduces an additional layer of averaging, weighted by the ratio $\lambda_i T/\left(\sum_{j=1}^N \lambda_j T\right) =\lambda_i /\left(\sum_{j=1}^N \lambda_j \right)$ that is the portion of the packets that arrive at the ingress node $s_i$ within $[t_0, t_0+T]$. The $\bar{D}_{\text{max}}$ in \eqref{eqn:metric_max} takes the maximum over all $\bar{D}_i$'s, which represents the highest average queueing delay of packets among all the ingress nodes.

We extend the definitions of both metrics to general multi-stage networks. We solely need to modify \eqref{eqn:bar_D_i} into 
$$
\bar{D}_i = \frac{1}{T}\int_{t_0}^{t_0+T} \sum_{p:~p[0]=n_i^1} w_{p} D_{p}(t) dt, ~\forall i=1,\dots,N_S
$$
where $\{p: p[0]=n_i^1$\} contains all the paths $p$ that start from the ingress node $n_i^1$, and $w_p$ represents the proportion of packets taking the path $p$ among all packets starting from $n_i^1$. We show in the proof of the results in Section \ref{subsec:general} that $w_p$ can be expressed as a function of the transmission rate vector $\mathbf{g}$. The formulation of $\bar{D}_{\text{avg}}$ and $\bar{D}_{\text{max}}$ are the same as \eqref{eqn:metric_avg} and \eqref{eqn:metric_max} respectively.






\section{Static Min-Delay Policy Design}
\label{sec:static_delay_optimal}

In this section, we develop the static policies that minimize $\bar{D}_{\text{avg}}$ and $\bar{D}_{\text{max}}$. We start from $N\times 1$ single-hop networks with $N_S=N$ ingress nodes and $N_D=1$ egress node. Examples include a single server receiving requests from multiple sources, and packets that arrive from multiple upstream links sharing a single port of a downstream switch between two stages in a data center \cite{singh2015jupiter}. We prove a sufficient and necessary condition on link rates that minimize both the delay metrics. The condition requires that the link rates from all the ingress nodes to the egress node are in the same proportion to their corresponding packet arrival rates. We term any policy under which the condition holds as a \emph{rate-proportional} policy. We unveil a counter-intuitive corollary that using larger link rates may increase delay. We then demonstrate that the rate-proportional policy can be extended to general single-hop networks and multi-stage networks, which guarantees minimum $\bar{D}_{\text{avg}}$ and $\bar{D}_{\text{max}}$.


\subsection{$N\times 1$ Networks}


Consider an $N\times 1$ network as shown in Fig.~\ref{fig:Nx1}. We identify the transmission rate vector $\mathbf{g}:=\{g_i\}_{i=1}^N$ that minimizes $\bar{D}_{\text{avg}}$ and $\bar{D}_{\text{max}}$, where we abbreviate link rate $g_{s_i,d}$ as $g_i$. We first derive the policies that minimize both delay metrics given unlimited link capacities, and then discuss how the capacities affect the result. 

\begin{figure}[!htbp]
\centering
\includegraphics[width=0.9\linewidth]{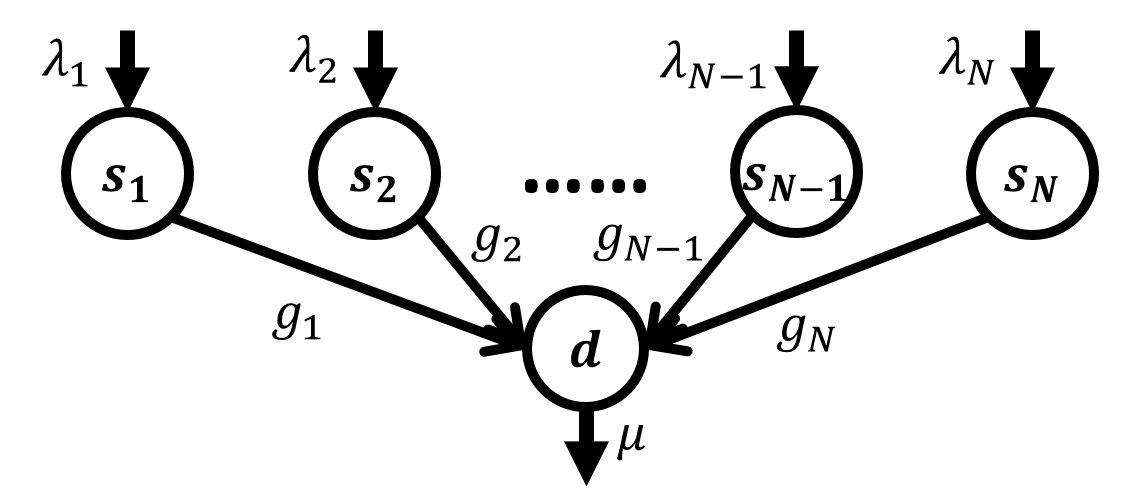}
\caption{An example of an $N\times 1$ single-hop network}
\label{fig:Nx1}
\end{figure}


We identify a sufficient and necessary condition on $\mathbf{g}$ that minimizes $\bar{D}_{\text{avg}}$ and $\bar{D}_{\text{max}}$ in Theorem \ref{thm:Nx1}. We give the detailed proof under $N=2$ for brevity, which can be extended to general $N$. We also derive the result under zero initial queue length for brevity. 

\begin{theorem}
\label{thm:Nx1}
Given an $N\times 1$ single-hop network with unlimited link capacity. For $\forall T>0$, the set of $\mathbf{g}=\{g_i\}_{i=1}^N$ that minimizes $\bar{D}_{\text{avg}}$ and $\bar{D}_{\text{max}}$ of the packets that arrive within $[t_0,t_0+T]$ where $\mathbf{q}(t_0)=\boldsymbol{0}$ is
\begin{equation}
\label{eqn:Nx1}
\begin{aligned}
& \left\{\left(\sum_{i=1}^N g_i\geq \mu \right)\cap \left(\frac{g_1}{\lambda_1}=\dots=\frac{g_N}{\lambda_N}\right)\right\} 
\\& \qquad \qquad  \cup \left\{g_i\geq \lambda_i,~\forall i=1,\dots,N\right\},
\end{aligned}
\end{equation}
under which $\bar{D}_{\text{avg}}=\bar{D}_{\text{max}}=\frac{T}{2\mu}\max\{\sum_{i=1}^N \lambda_i-\mu, 0\}$.
\end{theorem}


\begin{proof}
Consider $N=2$. The main idea of the proof is that we divide the feasible link rate region of $\mathbf{g}=(g_1,g_2)$, which is $[0,\infty)\times[0,\infty)$, into 4 sub-regions:
\begin{equation}
\label{eqn:R1_R4}
\begin{cases}
\mathcal{R}_1:=\{\mathbf{g}\mid g_1\in[0,\lambda_1],g_2\in[0,\lambda_2]\}\\
\mathcal{R}_2:=\{\mathbf{g}\mid g_1\in[\lambda_1,\infty),g_2\in[\lambda_2,\infty)\}\\
\mathcal{R}_3:=\{\mathbf{g}\mid g_1\in[\lambda_1,\infty),g_2\in[0,\lambda_2]\}
\\
\mathcal{R}_4:=\{\mathbf{g}\mid g_1\in[0,\lambda_1],g_2\in[\lambda_2,\infty)\}
\end{cases}
\end{equation}
and we identify the optimal $\mathbf{g}$'s restricted in each of these sub-regions, denoted by $\mathbf{g}_{(1)}^*,\mathbf{g}_{(2)}^*,\mathbf{g}_{(3)}^*,\mathbf{g}_{(4)}^*$ respectively. We show that each $\mathbf{g}_{(i)}^*$ leads to the same average queueing delay $\bar{D}_{\text{avg}}=\frac{T}{2\mu}\max\{(\lambda_1+\lambda_2-\mu),0\}$ and the same maximum ingress delay $\bar{D}_{\text{max}}=\frac{T}{2\mu}\max\{(\lambda_1+\lambda_2-\mu),0\}$. 

We define $D_{i}(t)$ as the total queueing delay of a packet injected into $s_i$ at time $t$. According to \eqref{eqn:two-hop-delay}, for $i=1,2$,
$$
\begin{aligned}
D_{i}(t)
&=\frac{q_{s_i}(t)}{g_i}+\max\left\{\frac{q_d(t)+ \frac{q_{s_i}(t)}{g_i}(g_1+g_2-\mu)}{\mu},0\right\}
\\&=
\begin{cases}
\frac{1}{\mu} \left(q_d(t)+ \frac{q_{s_i}(t)}{g_i}(g_1+g_2)\right), \qquad g_1+g_2\geq\mu \\
\frac{q_{s_i}(t)}{g_i}, \qquad g_1+g_2<\mu
\end{cases}
\end{aligned}
$$
due to $\mathbf{q}(t_0) = \boldsymbol{0}$ which guarantees that when $g_1+g_2<\mu$, $q_d(t)$ will keep zero and thus the only queueing delay is at the ingress nodes.
The average delay for packets that arrive to ingress node $s_i$ within $[t_0,t_0+T]$ is 
$
\bar{D}_i = \frac{1}{T}\int_{t_0}^{t_0+T} D_{i}(t) dt,~i=1,2,
$
based on which we can formulate the delay metrics $\bar{D}_{\text{avg}}$ given by \eqref{eqn:metric_avg} and $\bar{D}_{\text{max}}$ given by \eqref{eqn:metric_max} as functions of $\mathbf{g}$.


\textbf{Case 1: $\mathcal{R}_1:=\{\mathbf{g}\mid g_1\in[0,\lambda_1],g_2\in[0,\lambda_2]\}$}

In $\mathcal{R}_1$, we first consider the case when $g_1+g_2\geq \mu$.
$$
\small
\begin{aligned}
\bar{D}_{i}&:=\frac{1}{T}\int_{t_0}^{t_0+T}D_i(t) dt = \frac{1}{T}\int_{t_0}^{t_0+T} \frac{q_d(t)+ \frac{q_{s_i}(t)}{g_i}(g_1+g_2)}{\mu} dt
\\&=\frac{1}{T\mu}\int_{t_0}^{t_0+T} (t-t_0)\max\{g_1+g_2-\mu,0\}
\\&\quad \qquad + \frac{g_1+g_2}{g_i}(t-t_0)\max\{\lambda_i-g_i,0\} dt
\\&=\frac{T}{2\mu} \left(\lambda_1\frac{g_1+g_2}{g_i}-\mu\right),~i=1,2
\end{aligned}
$$
Then according to \eqref{eqn:metric_avg} and \eqref{eqn:metric_max},
$$
\footnotesize
\begin{aligned}
&\bar{D}_{\text{avg}}={\frac{\lambda_1}{\lambda_1+\lambda_2}} \bar{D}_{1}+{\frac{\lambda_2}{\lambda_1+\lambda_2}} \bar{D}_{2}\\&=\frac{T}{2\mu}\left(\frac{\lambda_1}{\lambda_1+\lambda_2} \left(\lambda_1\frac{g_1+g_2}{g_1}-\mu\right)+\frac{\lambda_2}{\lambda_1+\lambda_2} \left(\lambda_2\frac{g_1+g_2}{g_2}-\mu\right)\right)
\end{aligned}
$$
and
$$
\footnotesize
\begin{aligned}
\bar{D}_{\text{max}}&=\max\{\bar{D}_{1}, \bar{D}_{2}\}=\frac{T}{2\mu}\left\{\left(\lambda_1\frac{g_1+g_2}{g_1}-\mu\right),  \left(\lambda_2\frac{g_1+g_2}{g_2}-\mu\right) \right\}
\end{aligned}
$$

For $\bar{D}_{\text{avg}}$, we can obtain by Cauchy-Schwartz inequality that the optimal solutions are all $\mathbf{g}$ that satisfy
$
g_1+g_2\geq \mu,~ \frac{g_1}{g_2}=\frac{\lambda_1}{\lambda_2}
$
under which the average delay is
$\bar{D}_{\text{avg}} =\frac{T}{2\mu}(\lambda_1+\lambda_2-\mu)$. 
For $\bar{D}_{\text{max}}$, we can obtain that the set of $\mathbf{g}$'s that satisfy \eqref{eqn:Nx1} achieves the minimum $\bar{D}_{\text{max}}=\frac{T}{2\mu}(\lambda_1+\lambda_2-\mu)$.

We then consider the case when $g_1+g_2\leq\mu$. In this case there will be no queue backlog in the egress node, and thus
$$
\small
\begin{aligned}
\bar{D}_{i}&=\frac{1}{T}\int_{t_0}^{t_0+T} \frac{q_{s_i}(t)}{g_i} dt=\frac{\max\{\lambda_i-g_i,0\}}{g_iT}\int_{t_0}^{t_0+T} (t-t_0) dt
\\&=\frac{T}{2} \frac{\lambda_i-g_i}{g_i},~i=1,2.
\end{aligned}
$$
Therefore
$$
\begin{cases}
\bar{D}_{\text{avg}}
=\frac{T}{2(\lambda_1+\lambda_2)}\left(\frac{\lambda_1^2}{g_1}+\frac{\lambda_2^2}{g_2}-\lambda_1-\lambda_2\right) \\
\bar{D}_{\text{max}} = \frac{T}{2} \max\left\{\frac{\lambda_1-g_1}{g_1}, \frac{\lambda_2-g_2}{g_2}\right\}
\end{cases}
$$
Then under $g_1+g_2\leq\mu$, the optimal metric values are $\bar{D}_{\text{avg}}=\bar{D}_{\text{max}}=\frac{T}{2\mu}(\lambda_1+\lambda_2-\mu)$, achieved only at $g_1=\frac{\lambda_1}{\lambda_1+\lambda_2}\mu,~g_2=\frac{\lambda_2}{\lambda_1+\lambda_2}\mu$, which is on the boundary $g_1+g_2=\mu$.

\textbf{Case 2: $\mathcal{R}_2:=\{\mathbf{g}\mid g_1\in[\lambda_1,\infty),g_2\in[\lambda_2,\infty)\}$}

When $g_1\geq \lambda_1$, $D_{1}(t)=\frac{q_{s_1}(t)}{g_1}+\frac{q_d\left(t+\frac{q_{s_1}(t)}{g_1}\right)}{\mu}=\frac{q_d(t)}{\mu}$
as $q_{s_1}(t)$ keeps $0$ since $q_{s_1}(t_0)=0$, which means the queueing delay only occurs at node $d$. Similar for $D_{2}(t)$. Since $\lambda_1+\lambda_2>\mu$, packets will accumulate at node $d$ and at time $t$, and 
$q_d(t)=(\lambda_1+\lambda_2-\mu)t$. Thus
$$
D_{1}(t)=D_{2}(t)=\frac{q_d(t)}{\mu}=\frac{\lambda_1+\lambda_2-\mu}{\mu}t.
$$
and $D_{s_1}=\frac{1}{T}\int_{t_0}^{t_0+T}D_{1}(t) dt=\frac{T}{2\mu}(\lambda_1+\lambda_2-\mu)=D_{s_2}$, and hence $\bar{D}_{\text{avg}}=\bar{D}_{\text{max}}=\frac{T}{2\mu}(\lambda_1+\lambda_2-\mu)$ for $\forall \mathbf{g}\in \mathcal{R}_2$.

\textbf{Case 3: $\mathcal{R}_3:=\{\mathbf{g}\mid g_1\in[\lambda_1,\infty),g_2\in[0,\lambda_2]\}$}

Based on the derivation in case 1 and 2 respectively, we have $D_{s_1}=\frac{T}{2\mu}(\lambda_1+g_2-\mu)$ as packets that arrive at $s_1$ only suffer from delay at $d$, and $D_{s_2}=\frac{T}{2\mu}\left(\frac{\lambda_1+g_2}{g_2}\lambda_2-\mu\right)$ where packets that arrive at $s_2$ suffer from delay at $s_2$ and $d$. We can verify easily that any optimal $\mathbf{g}\in \mathcal{R}_3$ that achieves minimum $\bar{D}_{\text{avg}}=\bar{D}_{\text{max}}=\frac{T}{2\mu}(\lambda_1+\lambda_2-\mu)$ should be guaranteed that $g_2=\lambda_2$ holds.

\textbf{Case 4: $\mathcal{R}_4:=\{\mathbf{g}\mid g_1\in[0,\lambda_1],g_2\in[\lambda_2,\infty)\}$} Similar to case 3, where any optimal $\mathbf{g}\in \mathcal{R}_4$ satisfies $g_1=\lambda_1$.
\end{proof}

We term \eqref{eqn:Nx1} as the \emph{min-delay region} of the transmission rate vector $\mathbf{g}$, which consists of a line segment connecting two points $\left\{\frac{\lambda_i}{\sum_{j=1}^N \lambda_j}\mu\right\}_{i=1}^N$ and $\left\{\lambda_i\right\}_{i=1}^N$ in an $N$-dimensional space, and a polytope of $\mathbf{g}$'s where $g_i\geq \lambda_i,~\forall i=1,\dots,N$. Theorem \ref{thm:Nx1} demonstrates that to achieve minimum delay, if there exists one link $(s_i,d)$ with transmission rate $g_i$ higher than the packet arrival rate $\lambda_i$ at $s_i$, then all the other links should be as well; if on the contrary $g_i\leq \lambda_i$ for some $i$, then we need to guarantee that the link rates should be in the same proportion to the packet arrival rates among all the ingress nodes in order to achieve minimum delay, i.e. $g_i/\lambda_i,~\forall i=1,\dots,N$ are the same. We term any policy that satisfies this condition as a \emph{rate-proportional} policy. An important implication of this result is that setting link rates in the same proportion to the packet arrival rates can achieve minimum delay as done by setting them greater than the packet arrival rates but with less total bandwidth required. We give an example when $N=3$ in Fig.~\ref{fig:3x1_optimal_region_example}: {Both using the rate-proportional policy with $g_i \leq \lambda_i$ as Fig.~\ref{fig:3x1_optimal_region_example}(a) and setting $g_i \geq \lambda_i$, $i=1,2,3$ as Fig.~\ref{fig:3x1_optimal_region_example}(b) lead to minimum $\bar{D}_{\text{avg}}$ and $\bar{D}_{\text{max}}$, while the transmission rate vector in Fig.~\ref{fig:3x1_optimal_region_example}(c) is not in the min-delay region.} Moreover, the min-delay regions for both $\bar{D}_{\text{avg}}$ and $\bar{D}_{\text{max}}$ in an $N\times 1$ single-hop network are the same, which demonstrates that we can simultaneously achieve minimum average delay and in the meantime  balance the delay of packets injected into different ingress nodes. 

\begin{figure}[!htbp]
\centering
\includegraphics[width=1.0\linewidth]{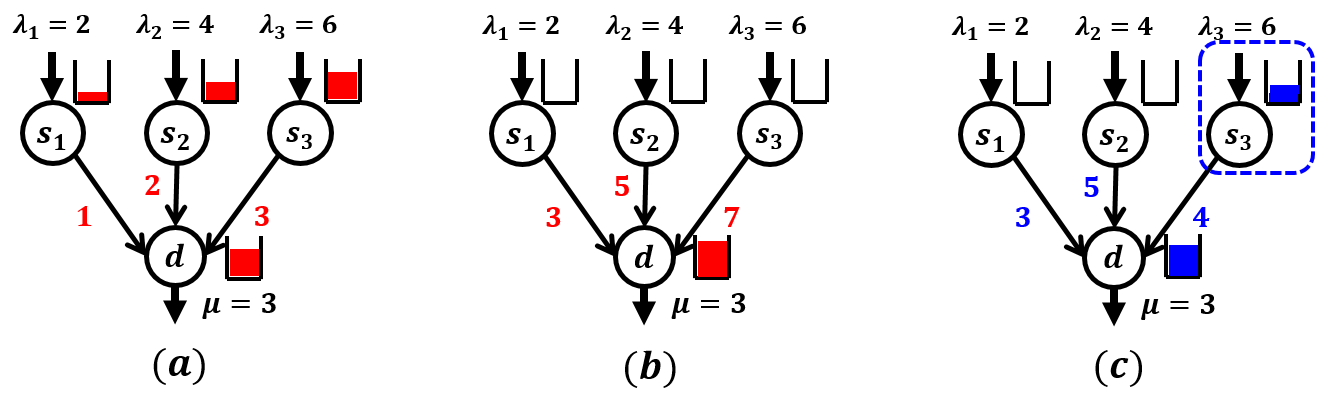}
\caption{A 3x1 example of Theorem \ref{thm:Nx1}: (a) Setting $g_i/\lambda_i=1/2, i=1,2,3$ satisfies \eqref{eqn:Nx1} and leads to minimum delay; (b) Setting $g_i\geq \lambda_i, i=1,2,3$ satisfies \eqref{eqn:Nx1} and leads to minimum delay, despite different queue growth rates compared with (a); (c) Setting $\mathbf{g}=\{3,5,4\}$ does not satisfy \eqref{eqn:Nx1} and thus does not incur minimum $\bar{D}_{\text{avg}}$ and $\bar{D}_{\text{max}}$, although all the 3 links rates are greater than those in (a).}
\label{fig:3x1_optimal_region_example}
\end{figure}

We visualize the min-delay region when $N=2$ for detailed explanation in Fig.~\ref{fig:result-2x1-switch}(a), which is marked as the orange area: a line segment connecting the points $\left(\frac{\lambda_1}{\lambda_1+\lambda_2}\mu, \frac{\lambda_2}{\lambda_1+\lambda_2}\mu\right)$ and $(\lambda_1,\lambda_2)$, and the polytope $\{\mathbf{g}\mid g_{i}\geq \lambda_i, i=1,2\}$. We mark the $\mathcal{R}_1$ to $\mathcal{R}_4$ in \eqref{eqn:R1_R4} in Fig.~\ref{fig:result-2x1-switch}(a). We have the following insights: (i) Setting $g_{i}$ no less than $\lambda_i$ for both $i=1,2$ achieves minimum queueing delay, while further increasing $g_{1}$ and $g_{2}$ does not make a difference. This is because for any $\mathbf{g}$ that $g_{i}\geq \lambda_i$, the buffers of $s_1$ and $s_2$ are empty, hence all the queueing delay is at the egress node $d$ bottlenecked by $\mu$. 
(ii) We can achieve the minimum delay in $\mathcal{R}_1$ using lower transmission rates compared with those in $\mathcal{R}_2$ by consider the rate-proportional policy where $\mathbf{g} = \{g_1, g_2\}$ satisfies
$g_{1}/g_{2} = \lambda_1/\lambda_2$, and meanwhile maximum throughput is guaranteed, i.e., $g_{1}+g_{2}\geq \mu$. 
The minimum total link rate is $\mu$ at the intersection point $\left(\frac{\lambda_1}{\lambda_1+\lambda_2}\mu, \frac{\lambda_2}{\lambda_1+\lambda_2}\mu\right)$.
(iii) It is not true that serving with higher rates leads to lower queueing delay. For example serving with transmission rates in $\mathcal{R}_3$ and $\mathcal{R}_4$ is inferior to controlling the transmission rates on the optimal line segment in $\mathcal{R}_1$. The counter-intuition is because packets from $s_1$ and $s_2$ share an egress node, where the imbalance between $g_{1}$ and $g_{2}$ leads to severe delay increase of packets that arrive to the ingress node with lower link rate downstream. 

\begin{figure}[!htbp]
\centering
\includegraphics[width=0.95\linewidth]{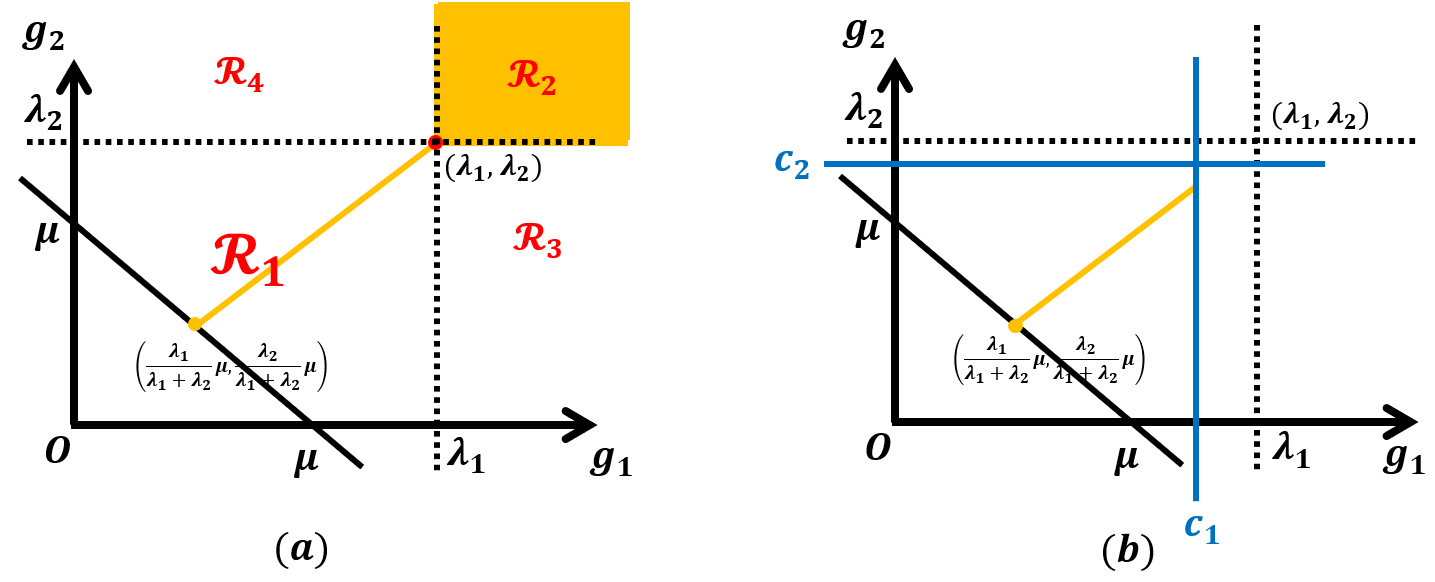}
\caption{The min-delay region in a $2\times 1$ single-hop network: (a) unlimited capacity; (b) limited capacity $(c_i\leq \lambda_i,~i=1,2)$}
\label{fig:result-2x1-switch}
\end{figure}

We further extend Theorem \ref{thm:Nx1} to the case of limited link capacities (i.e., capacity of link $(s_i, d)$ is $c_i$). We can obtain directly that the min-delay region for limited capacity case is simply the intersection of \eqref{eqn:Nx1} and $\{\mathbf{g}\mid g_{i}\leq c_{i}, i=1,\dots,N\}$
as limited capacity does not affect the proof, where $c_i$ is the abbreviation of $c_{s_id}$. We illustrate the min-delay region given that $\lambda_i>c_{i},~i=1,2$ in Fig.~\ref{fig:result-2x1-switch}(b), which is a single solid orange line segment. We observe that serving both links with maximum rates equal to the link capacity does not lead to minimum delay in general cases, which validates the necessity of refined control of link rates based on \eqref{eqn:Nx1}. This result also gives insights on the demand-aware bandwidth allocation in data center networks \cite{nagaraj2016numfabric, kumar2015bwe, zhang2021gemini}, where allocating bandwidth in proportion to the demands from different ingress nodes leads to minimum delay when overload occurs. The minimum total bandwidth $c_1 + c_2$ required to achieve the global minimum delay as in the case with unlimited capacity is $\mu$, where $c_i = \frac{\lambda_i}{\lambda_1+\lambda_2}\mu,~i=1,2$.


{
Finally we discuss the impact of the initial queue length $\mathbf{q}(t_0)$ on the result. Consider $N=2$ for example. We can follow the proof of Theorem \ref{thm:Nx1} and obtain the min-delay region in $\mathcal{R}_1$ to be
\begin{equation}
\label{eqn:opt_init_queue}
 \frac{g_1}{g_2} = \sqrt{\frac{\lambda_1(\lambda_1+q_{s_1}(t_0)/T)}{\lambda_2(\lambda_2+q_{s_2}(t_0)/T)}},   
\end{equation}
under which the queue length at any ingress node will not reduce to zero. Note that \eqref{eqn:opt_init_queue} also follows the rate-proportional pattern with initial queue length and duration $T$ included.
For $\mathcal{R}_2$, $\mathcal{R}_3$ and $\mathcal{R}_4$, the derivation is of higher complexity as we need to analyze if the queue length at the ingress nodes will or will not change from non-zero to zero within $[t_0,t_0+T]$, which involves at least two cases for each ingress node. In practice, the initial queue length $q_{s_1}(t_0)$ and $q_{s_2}(t_0)$ are generally very small before overload occurs, and we generally care about rate control for relatively long $T$ instead of instantaneous overload. Therefore $q_{s_1}(t_0)/T$ is generally small and thus \eqref{eqn:opt_init_queue} is approximately $g_1/g_2=\lambda_1/\lambda_2$, matching \eqref{eqn:Nx1}. For these reasons and the conciseness of proof, we neglect initial queue length and verify empirically in Section \ref{sec:evaluation} that initial queue length does not affect the overall performance.
}

\subsection{General Single-Hop and Multi-Stage Networks}
\label{subsec:general}

We extend the {rate-proportional} policy shown in Theorem \ref{thm:Nx1} to general single-hop networks and multi-stage networks and show that it is a sufficient condition for queueing delay minimization. {The extended rate-proportional policy requires that all the nodes at the same  layer share the same ratio between their total ingress rates and egress rates of packets.}

\subsubsection{$N_S\times N_D$ single-hop networks} We derive a sufficient condition on $\mathbf{g}$ to achieve minimum $\bar{D}_{\text{avg}}$ and $\bar{D}_{\text{max}}$ in Theorem \ref{thm:NxM} given unlimited capacity in $N_S\times N_D$ single-hop networks. We can extend the result to the case of limited capacity by adding the constraints $g_{ij} \leq c_{ij}, \forall (i,j)\in \mathcal{E}$ as discussed in $N\times 1$ networks.

\begin{theorem}
\label{thm:NxM}
Given an $N_S\times N_D$ single-hop network with unlimited link capacity. For $\forall T>0$, a sufficient condition to globally minimize both $\bar{D}_{\text{avg}}$ and $\bar{D}_{\text{max}}$ of the packets that arrive within $[t_0,t_0+T]$ where $\mathbf{q}(t_0)=\boldsymbol{0}$ is
\begin{equation}
\label{eqn:optimum_NxM}
\begin{cases}
\frac{\sum_{k=1}^{N_D} g_{s_id_k}}{\sum_{k=1}^{N_D} g_{s_jd_k}}=\frac{\lambda_i}{\lambda_j}, \forall i,j=1,\dots,N_S \\
\\
\frac{\sum_{k=1}^{N_S} g_{s_kd_i}}{\sum_{k=1}^{N_S} g_{s_kd_j}}=\frac{\mu_i}{\mu_j}, \forall i,j=1,\dots,N_D \\
\\
\sum_{k=1}^{N_S} g_{s_kd_j} \geq \mu_j, \forall j = 1,\dots,N_D
\end{cases}
\end{equation}
under which $\bar{D}_{\text{avg}}=\bar{D}_{\text{max}}=\frac{T}{2\sum_{j=1}^{N_D} \mu_j}\max\{\sum_{i=1}^{N_S} \lambda_i -\sum_{j=1}^{N_D} \mu_j,0\}$. Furthermore, 
\eqref{eqn:optimum_NxM} is both sufficient and necessary for minimizing $\bar{D}_{\text{avg}}$ and $\bar{D}_{\text{max}}$ over the policies in $\{\mathbf{g}\mid \sum_{j=1}^{N_D} g_{s_id_j}\leq \lambda_i,~\forall i=1,\dots,N_S\}$.
\end{theorem}

We defer the proof to the appendix. We explain the min-delay conditions \eqref{eqn:optimum_NxM}: The first constraint requires that the total egress rates of different ingress nodes should be in the same proportion to their packet arrival rates $\{\lambda_i\}_{i=1}^{N_S}$; The second constraint requires that the total ingress rates of different egress nodes should be in the same proportion to their service rates $\{\mu_j\}_{j=1}^{N_D}$; The third constraint guarantees maximum throughput. Any transmission policy that satisfies these three conditions guarantees minimum $\bar{D}_{\text{avg}}$ and $\bar{D}_{\text{max}}$. Compared with \eqref{eqn:Nx1} that requires the rate-proportional property at the ingress layer solely for $N\times 1$ networks, \eqref{eqn:optimum_NxM} requires the rate-proportional property at both the ingress and egress layers for general single-hop networks. 

We further point out that \eqref{eqn:optimum_NxM} is also a necessary condition for delay minimization under limited transmission rates where $\mathbf{g} \in \{\mathbf{g}\mid \sum_{j=1}^{N_D} g_{s_id_j}\leq \lambda_i,~\forall i=1,\dots,N_S\}$, i.e., the egress rate of node $s_i$ is no greater than the packet arrival rate $\lambda_i$. This result indicates that given limited link capacity, the rate-proportional policies are the only ones that minimize both delay metrics in overloaded single-hop networks.

We give a $2\times 2$ example in Fig.~\ref{fig:2x2_optimal_region_example}, which demonstrates that multiple solutions can minimize the delay metrics, as long as they satisfy \eqref{eqn:optimum_NxM} shown in Fig.~\ref{fig:2x2_optimal_region_example}(a) and (b), while higher link rates may even increase the delay shown in Fig.~\ref{fig:2x2_optimal_region_example}(c).

\begin{figure}[!htbp]
\centering
\includegraphics[width=1.0\linewidth]{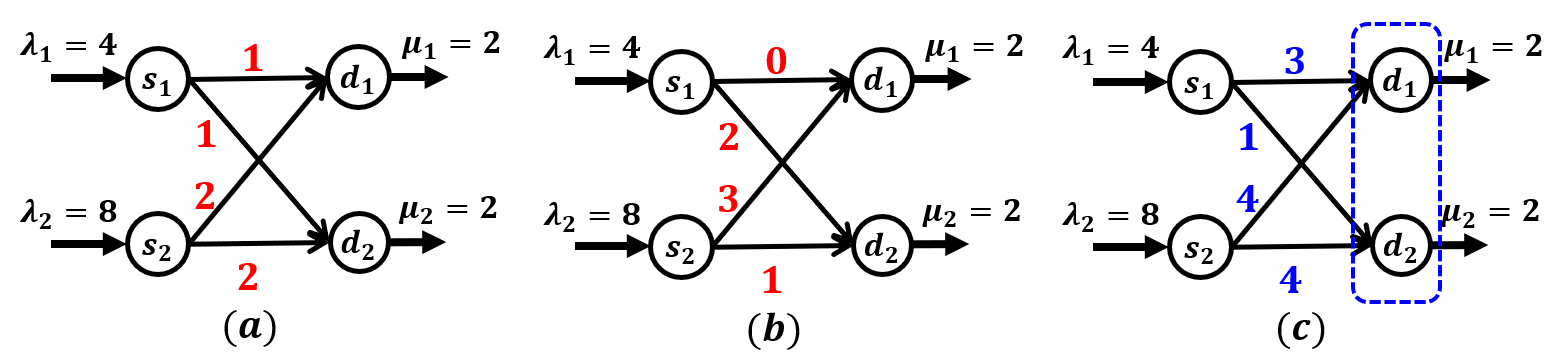}
\caption{A 2x2 example of Theorem \ref{thm:NxM}: (a) and (b) set $\mathbf{g}$ that satisfies $(g_{s_1d_1} + g_{s_1d_2}) / (g_{s_2d_1} + g_{s_2d_2}) = \lambda_1 / \lambda_2$ and  $(g_{s_1d_1} + g_{s_2d_1}) / (g_{s_1d_2} + g_{s_2d_2}) = \mu_1 / \mu_2$ which lead to minimum $\bar{D}_{\text{avg}}$ and $\bar{D}_{\text{max}}$ simultaneously; Setting link rates as in (c) does not lead to minimum delay since $(g_{s_1d_1} + g_{s_2d_1}) / (g_{s_1d_2} + g_{s_2d_2}) \neq \mu_1 / \mu_2$, although the total link rates are higher than (a) and (b).}
\label{fig:2x2_optimal_region_example}
\end{figure}






\subsubsection{Multi-stage networks} We further extend the min-delay conditions for single-hop networks to multi-stage networks with $L$ layers. We show in Theorem \ref{thm:clos-static-flow} that applying the rate-proportional policy design over each of the $L$ layers leads to minimum $\bar{D}_{\text{avg}}$ and $\bar{D}_{\text{max}}$ as long as the maximum throughput is guaranteed. 

\begin{theorem}
\label{thm:clos-static-flow}
Consider an $L$-layer network with unlimited link capacity. 
For $\forall T>0$, a sufficient condition to globally minimize both $\bar{D}_{\text{avg}}$ and $\bar{D}_{\text{max}}$ of the packets that arrive within $[t_0,t_0+T]$ when $\mathbf{q}(t_0)=\boldsymbol{0}$ is
\begin{equation}
\label{eqn:clos-static-flow}
\begin{cases}
\frac{\lambda_{i}}{\sum_{n_j^2\in \mathcal{V}_2} g_{n_i^1,n_j^2}} 
= \gamma_1, \quad \forall n_i^1 \in \mathcal{V}_1 \\
\\
\frac{\sum_{n_k^{l-1}\in \mathcal{V}_{l-1}} g_{n_k^{l-1}, n_i^l}}{\sum_{n_j^{l+1} \in \mathcal{V}_{l+1}} {g}_{n_i^l, n_j^{l+1}}} = \gamma_l, \quad \forall n_i^l \in \mathcal{V}_l, ~\forall l = 2,\dots,L-1 \\
\\
\frac{\sum_{n_k^{L-1}\in \mathcal{V}_{L-1}} {g}_{n_k^{L-1}, n_i^L}}{\mu_i} = \gamma_L, \quad \forall n_i^L \in \mathcal{V}_L
\end{cases}
\end{equation}
for some $\boldsymbol{\gamma}=\{\gamma_l\}_{l=1}^L \in \mathbb{R}_+^{L}$ and the maximum throughput is achieved, where $\bar{D}_{\text{avg}}= \bar{D}_{\text{max}} = \frac{T}{2}\max\left\{\frac{\sum_{i=1}^{N_S} \lambda_i}{\sum_{j=1}^{N_D} \mu_j} - 1, 0\right\}$.
\end{theorem}

We defer the proof to the appendix, whose primary idea is to apply $L-1$ times the proof idea of Theorem \ref{thm:NxM} for single-hop networks. Theorem \ref{thm:clos-static-flow} shows that we guarantee minimum delay given that maximum throughput is achieved by setting link rates such that at each layer $l$, the ingress rates of all the nodes are in the same proportion to their egress rates. We denote the ratio between the ingress and egress rates at nodes at the $l$-th layer by $\gamma_l$. Note that we do not need to guarantee $\gamma_i = \gamma_j$ for different layers $i$ and $j$. We give an example of $\mathbf{g}$ that satisfies \eqref{eqn:clos-static-flow} in Fig. \ref{fig:Paper_multi-layer_example}. An important implication of the explicit sufficient condition is that they simplify the problem of link rate control for delay minimization to finding a feasible solution to \eqref{eqn:clos-static-flow} given a feasible $\boldsymbol{\gamma}$, which can be formulated as a linear programming (LP) problem. We defer the discussion of its feasibility analysis and its wider applications to Section \ref{subsec:practical}. We point out that it is challenging to derive the necessary condition for multi-stage networks and we leave it to future work. 

\begin{figure}[!htbp]
\centering
\includegraphics[width=0.95\linewidth]{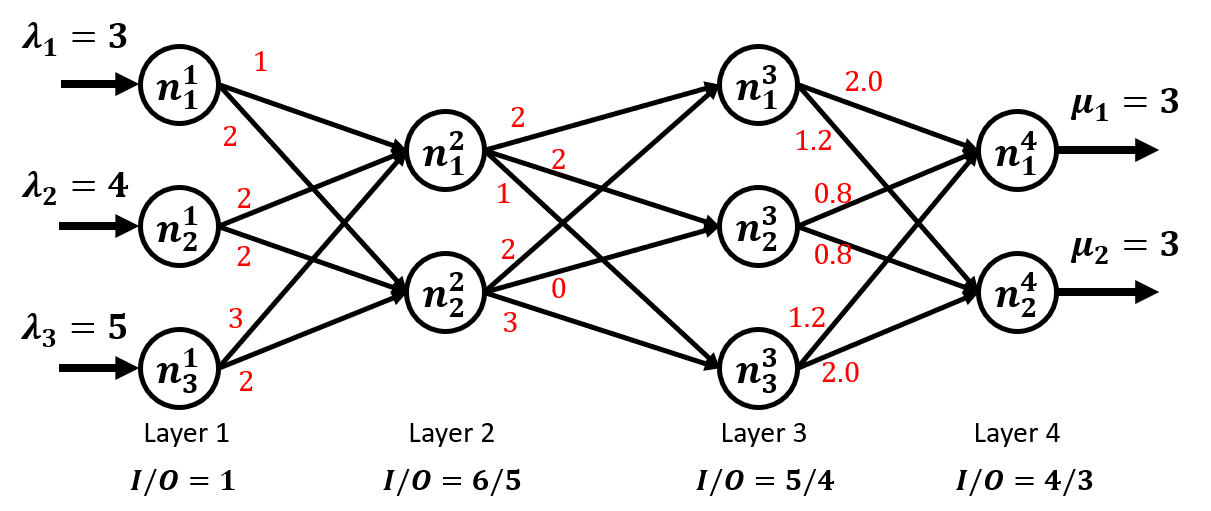}
\caption{An example of link rate control in a 4-layer multi-stage network which minimizes $\bar{D}_{\text{avg}}$ and $\bar{D}_{\text{max}}$. The red numbers are link rate values. The ingress and egress rates of nodes at layer 1 to 4 are $1$, $6/5$, $5/4$, and $4/3$ respectively.}
\label{fig:Paper_multi-layer_example}
\end{figure}

We leave a final remark to conclude this section. Note that the rate-proportional policy design is not equivalent to overload balancing in the buffers of nodes at the same layer  \cite{wu2022overload, li2014dynamic}. The former is to maintain identical ratios between ingress and egress rates, while the latter is to maintain identical differences between them. This means that the delay minimization and overload balancing among node buffers cannot be simultaneously achieved when the network is overloaded in general.


\section{Queue-based Min-delay Policy Design}
\label{sec:queue_policy}

In this section, we develop queue-based dynamic policies where link rates can be adjusted according to real-time queue backlog information in the network, based on the static min-delay policy design from Section \ref{sec:static_delay_optimal}. The motivation to study queue-based min-delay policies is that the static policies require the complete knowledge of network parameters, which in real networks may be difficult to estimate or unavailable, for example the packet arrival rate vector $\boldsymbol{\lambda}$ \cite{neely2010stochastic}, 
while the real-time queue backlog $\mathbf{q}(t)$ is often accessible. 
We demonstrate that the \emph{queue-proportional} policies can achieve minimum queueing delay asymptotically: setting the egress link rates of nodes at a layer in the same proportion to the queue backlog length in their buffers. 
We first introduce the min-delay queue-based policy design in $N\times 1$ networks, and then extend it to general single-hop and multi-stage networks.

\subsection{$N\times 1$ Networks}

We propose the queue-based min-delay policy for $N\times 1$ networks (as in Fig.~\ref{fig:Nx1}) based on the static counterparts from Theorem \ref{thm:Nx1}.
The min-delay condition \eqref{eqn:Nx1} on static $\mathbf{g}$ which requires $g_i/g_j=\lambda_i/\lambda_j, ~\forall i \neq j$ implies that the dynamic rate control policy where $\mathbf{g}(t)$ satisfies  
$g_i(t)/g_j(t)=\dot{q}_{s_i}(t)/\dot{q}_{s_j}(t), ~\forall i \neq j$
can minimize both $\bar{D}_{\text{avg}}$ and $\bar{D}_{\text{max}}$, since
$$
\frac{g_i(t)}{g_j(t)}=\frac{\dot{q}_{s_i}(t)}{\dot{q}_{s_j}(t)}=\frac{\lambda_i-g_i(t)}{\lambda_j-g_j(t)}\overset{(\ast)}{=}\frac{\lambda_{i}}{\lambda_j}
$$
where $(\ast)$ holds since if for some $a,b,c,d\neq 0$ and $a+c, b+d\neq 0$, $a/b=c/d$, then $a/b=c/d=(a+c)/(b+d)$. This dynamic policy inspires the following queue-based policy without utilizing the information of $\boldsymbol{\lambda}$:
\begin{equation}
\label{eqn:queue_policy_Nx1}
\frac{g_{i}(\mathbf{q}(t))}{g_{j}(\mathbf{q}(t))}=\frac{q_{i}(t)}{q_{j}(t)}, ~\forall i \neq j,\quad 
\sum_{i=1}^N g_{i}(\mathbf{q}(t)) \geq \mu
\end{equation}
where the link rates are set in the same proportion to the real-time queue length at the ingress nodes. We term any policy that follows \eqref{eqn:queue_policy_Nx1} as a queue-proportional policy. We show in Theorem \ref{thm:optimum_queue_Nx1} that \eqref{eqn:queue_policy_Nx1} achieves optimal $\bar{D}_{\text{avg}}$ and $\bar{D}_{\text{max}}$ with zero initial queue length at all nodes, and further in Theorem \ref{thm:asymp_queue_Nx1} that \eqref{eqn:queue_policy_Nx1} asymptotically converges to the min-delay policy \eqref{eqn:Nx1} given arbitrary initial queue vector $\mathbf{q}(t_0)$.

\begin{theorem}
\label{thm:optimum_queue_Nx1}
With $\mathbf{q}(t_0)=\boldsymbol{0}$, then the policy \eqref{eqn:queue_policy_Nx1} achieves minimum $\bar{D}_{\text{avg}}$ and $\bar{D}_{\text{max}}$ as in \eqref{eqn:Nx1}.
\end{theorem}

\begin{proof}
Initially at $t_0$, we take $\epsilon\rightarrow 0$ and have $\forall i \neq j$, 
$$
\small
\begin{aligned}
&\frac{g_i(\mathbf{q}(t_0+\epsilon))}{g_j(\mathbf{q}(t_0+\epsilon))}=\frac{q_{s_i}(t_0+\epsilon)}{q_{s_j}(t_0+\epsilon)}
=\frac{\int_{t_0}^{t_0+\epsilon} \lambda_i-g_i(\mathbf{q}(s))ds}{\int_{t_0}^{t_0+\epsilon} \lambda_j-g_j(\mathbf{q}(s))ds}
\\&= \frac{\int_{t_0}^{t_0+\epsilon} \lambda_i ds - g_i(\mathbf{q}(t_0+\alpha_i\epsilon))}{\int_{t_0}^{t_0+\epsilon} \lambda_j ds- g_j(\mathbf{q}(t_0+\alpha_j\epsilon))}\rightarrow \frac{\int_{t_0}^{t_0+\epsilon} \lambda_i ds}{\int_{t_0}^{t_0+\epsilon} \lambda_j ds}  =\frac{\lambda_i}{\lambda_j}
\end{aligned}
$$
where $\alpha_i, \alpha_j \in [0,1]$. Then in the time interval $[t_0+\epsilon,t_0+2\epsilon]$,
$$
\small
\begin{aligned}
\frac{g_i(\mathbf{q}(t_0+2\epsilon))}{g_j(\mathbf{q}(t_0+2\epsilon))}&=\frac{q_{s_i}(t_0+2\epsilon)}{q_{s_j}(t_0+2\epsilon)}
=\frac{q_{s_i}(t_0+\epsilon)+\epsilon\dot{q}_{s_i}}{q_{s_j}(t_0+\epsilon)+\epsilon\dot{q}_{s_j}}
\\&=\frac{q_{s_i}(t_0+\epsilon)+\epsilon(\lambda_i-g_i(\mathbf{q}(t_0+\epsilon)))}{q_{s_j}(t_0+\epsilon)+\epsilon(\lambda_j-g_j(\mathbf{q}(t_0+\epsilon)))}
=\frac{\lambda_i}{\lambda_j}
\end{aligned}.
$$
Iteratively, we can obtain
$$
\small
\frac{g_{i}(\mathbf{q}(t))}{g_{j}(\mathbf{q}(t))}=\frac{q_{s_i}(t)}{q_{s_j}(t)}=\frac{\lambda_i}{\lambda_j}, ~\forall t
$$
which minimizes $\bar{D}_{\text{avg}}$ and $\bar{D}_{\text{max}}$ at any time $t$ according to Theorem \ref{thm:Nx1}.
\end{proof}

\begin{theorem}
\label{thm:asymp_queue_Nx1}
With arbitrary $\mathbf{q}(t_0)$, \eqref{eqn:queue_policy_Nx1} converges to the state where $\lim_{t\rightarrow \infty} {g_i(\mathbf{q}(t))}/{g_j(\mathbf{q}(t))}=\lambda_i/\lambda_j, ~\forall i \neq j$ which minimizes $\bar{D}_{\text{avg}}$ and $\bar{D}_{\text{max}}$.
\end{theorem}

\begin{proof}
Under \eqref{eqn:queue_policy_Nx1}, for any $i\neq j$, when $t\rightarrow \infty$,
$$
\small
\begin{aligned}
\left|\frac{q_{s_i}(t)}{q_{s_j}(t)} - \frac{\lambda_i}{\lambda_j}\right| &= \left|\frac{q_{s_i}(t_0)+\int_{t_0}^t \lambda_i - g_i(\mathbf{q}(s)) ds}{q_{s_j}(t_0)+\int_{t_0}^t \lambda_j - g_j(\mathbf{q}(s)) ds} - \frac{\lambda_i}{\lambda_j}\right|
\\& \overset{\text{L'hos}}{=} \left|\frac{\lambda_i - g_i(\mathbf{q}(t))}{\lambda_j - g_j(\mathbf{q}(t))} - \frac{\lambda_i}{\lambda_j}\right| 
\end{aligned}
$$
and 
\begin{equation}
\label{eqn:l_hos}
\frac{g_{i}(\mathbf{q}(t))}{g_{j}(\mathbf{q}(t))}=\frac{q_{s_i}(t)}{q_{s_j}(t)} \overset{\text{L'hos}}{=}  \frac{\dot{q}_{s_i}(t)}{\dot{q}_{s_j}(t)} = \frac{\lambda_i - g_i(\mathbf{q}(t))}{\lambda_j - g_j(\mathbf{q}(t))}
\end{equation}
where L'hos means applying the L'hospital's rule\footnote{We solely need to consider the case where the queue backlogs in ingress nodes keep growing (i.e, $\lim_{t\rightarrow \infty}\int_{t_0}^t \lambda_i - g_i(\mathbf{q}(s)) ds \rightarrow \infty,~\forall i=1,\dots,N_S$), since otherwise the min-delay condition is trivial by serving with link rates higher than packet arrival rates at all the ingress nodes, as shown in Theorem \ref{thm:Nx1}.}. We further derive based on \eqref{eqn:l_hos} that
$
\lim_{t\rightarrow \infty} \frac{q_{s_i}(t)}{q_{s_j}(t)} =  \lim_{t\rightarrow \infty} \frac{g_{i}(\mathbf{q}(t))}{g_{j}(\mathbf{q}(t))} =  \frac{\lambda_i - \lim_{t\rightarrow \infty} g_i(\mathbf{q}(t))}{\lambda_j - \lim_{t\rightarrow \infty} g_j(\mathbf{q}(t))} = \frac{\lambda_i}{\lambda_j}.
$ 
This shows that the transmission rates asymptotically become proportional to the corresponding packet arrival rates at the ingress nodes, which reaches minimum $\bar{D}_{\text{avg}}$ and $\bar{D}_{\text{max}}$ based on Theorem \ref{thm:Nx1}. 
\end{proof}

Although not necessarily achieving minimum delay at any time given arbitrary $\mathbf{q}(t_0)$, the policy \eqref{eqn:queue_policy_Nx1} keeps driving the queueing dynamics to the state under which delay is minimized. Intuitively, it drives $q_i(t) / q_j(t) \rightarrow \lambda_i / \lambda_j$: Suppose $q_i(t) / q_j(t) > \lambda_i / \lambda_j$, then $g_i(\mathbf{q}(t)) / g_j(\mathbf{q}(t)) > \lambda_i / \lambda_j$ which drives down $q_i(t) / q_j(t)$, and when $q_i(t) / q_j(t) < \lambda_i / \lambda_j$, the policy increases $q_i(t) / q_j(t)$ closer to $\lambda_i / \lambda_j$. 


\subsection{General Single-Hop and Multi-Stage Networks}

%

We extend the asymptotic min-delay policy \eqref{eqn:queue_policy_Nx1} to general single-hop and multi-stage networks. We show that adjusting link rates so that the egress rates of nodes at the same layer in the same proportion to their queue backlogs at every timestamp leads to minimum delay asymptotically. Theorem \ref{thm:optimum_queue_NxM} delivers a sufficient condition that minimizes $\bar{D}_{\text{avg}}$ and $\bar{D}_{\text{max}}$ asymptotically in  single-hop networks. Theorem \ref{thm:clos-queue-based-policy} further generalizes the result to multi-stage networks.
\begin{theorem}
\label{thm:optimum_queue_NxM}
Consider an $N_S\times N_D$ single-hop network. Any queue-based policy $\mathbf{g}(\mathbf{q}(t)),~\forall t$ that satisfies 
\begin{equation}
\label{eqn:queue_policy_NxM}
\small
\begin{cases}
\frac{\sum_{k=1}^{N_D} g_{s_id_k}(\mathbf{q}(t))}{\sum_{k=1}^{N_D} g_{s_jd_k}(\mathbf{q}(t))}=\frac{q_{s_i}(t)}{q_{s_j}(t)},~\forall i,j=1,\dots,N_S \\
\\
\frac{\sum_{k=1}^{N_S} g_{s_kd_i}(\mathbf{q}(t))}{\sum_{k=1}^{N_S} g_{s_kd_j}(\mathbf{q}(t))}=\frac{\mu_i}{\mu_j},~\forall i,j=1,\dots,N_D \\
\\
\sum_{k=1}^{N_S} g_{s_kd_j}(\mathbf{q}(t))\geq \mu_j,~\forall j = 1,\dots,N_D
\end{cases}
\end{equation}
achieves 
asymptotically minimum $\bar{D}_{\text{avg}}$ and $\bar{D}_{\text{max}}$ as in \eqref{eqn:optimum_NxM} with arbitrary initial queue backlog.
\end{theorem}

\begin{theorem}
\label{thm:clos-queue-based-policy}
Consider an $L$-layer multi-stage network. Any queue-based policy $\mathbf{g}(\mathbf{q}(t)), \forall t$ that satisfies
\begin{equation}
\small
\label{eqn:clos-queue-based}
\begin{cases}
\frac{q_{n_i^l}(t)}{\sum_{n_j^{l+1} \in \mathcal{V}_{l+1}} {g}_{n_i^l,n_j^{l+1}}(\mathbf{q}(t))}  = \gamma_l, \quad \forall n_i^l \in \mathcal{V}_l, ~\forall l = 1,\dots,L-1 \\
\\
\frac{\sum_{n_k^{L-1} \in \mathcal{V}_L} g_{n_k^{L-1}, n_i^L}(\mathbf{q}(t))}{\mu_i} = \gamma_L, \quad \forall n_i^L \in \mathcal{V}_L
\end{cases}
\end{equation}
for some $\boldsymbol{\gamma} = \{\gamma_l\}_{l=1}^L \in \mathbb{R}_+^L$ and guarantees maximum throughput can achieve asymptotic minimum $\bar{D}_{\text{avg}}$ and $\bar{D}_{\text{max}}$ with arbitrary initial queue backlog.
\end{theorem}

The proof idea of both theorems follows Theorem  \ref{thm:asymp_queue_Nx1}. Theorem \ref{thm:optimum_queue_NxM} is a special case of Theorem \ref{thm:clos-queue-based-policy} with $L=2$. Both demonstrate the idea of min-delay link rate control by maintaining the  egress rates of all the nodes in the same proportion to the current queue backlogs in these nodes at the same layer. 

In summary, we determined that the format of the min-delay queue-based policy is similar to the static policies. The above analysis reduces the optimization problem in dynamical systems to finding feasible solutions to an explicit set of queue-proportional constraints. We show below in Section \ref{subsec:practical} that the explicit form facilitates the co-optimization of queueing delay together with other metrics. 



\section{Performance Evaluation}
\label{sec:evaluation}

In this section, we evaluate the proposed min-delay policies in overloaded networks. We compare the performance of $\bar{D}_{\text{avg}}$ and $\bar{D}_{\text{max}}$ of our proposed methods with (i) the \emph{Max-link-rate} policy where all links are activated with rates equal to their capacities, and (ii) the \emph{Backpressure} policy that achieves optimal throughput and low latency \cite{georgiadis2006optimal}, which serves packets over a link $(i,j)$ with rate equal to its capacity if and only if node $i$ has longer queue backlog than node $j$. We use the following abbreviations in the evaluation: OPT for the proposed min-delay policies, MAX for the max-link-rate policy, and BP for the backpressure policy.

We validate that our proposed methods achieve minimum $\bar{D}_{\text{avg}}$ and $\bar{D}_{\text{max}}$ in various network settings: (i) different topologies of both single-hop and multi-stage networks; (ii) different values of $\boldsymbol{\lambda}$ and $\boldsymbol{\mu}$; (iii) different capacity vectors $\mathbf{c}$ (single-hop networks only). 
We consider multiple network instances with randomly sampled values of the above parameters, and measure the empirical cumulative distribution functions (CDFs) of $\bar{D}_{\text{avg}}$ and $\bar{D}_{\text{max}}$. We solely present the results for the queue-based policies, where the static policies result in similar performance. Moreover, packets are transmitted in discrete time intervals during simulation, and the results demonstrate that the min-delay property of our proposed policies based on the continuous fluid model holds under discrete transmission.

\subsection{$N\times 1$ Networks}

We evaluate the delay performance over $32\times 1$ single-hop networks. We consider $500$ different combinations of parameter settings sampled based on the following rules: (i) The arrival rate $\lambda_i$ to each ingress node $s_i$ is uniformly distributed in $[12,20]$; (ii) The service rate of the shared egress node is $0.4\times \sum_{i=1}^{32}\lambda_i$ so that the network is overloaded; (iii) Link capacities are uniformly distributed within $[20,35]$ to represent the case of sufficient capacity where $\lambda_i \leq c_i$ for each node $s_i$, and $[5,15]$ to represent the case of limited capacity where $\lambda_i$ may exceed $c_i$.
We round any rational number to an integer to characterize discrete packet transmission. We consider the initial queue length in each node to be a random integer within $[101,300]$, and we consider the $\bar{D}_{\text{avg}}$ and $\bar{D}_{\text{max}}$ of packets that arrive within the first $200$ time units that the network is overloaded. 

Fig.~\ref{fig:32x1_suff} illustrates the CDF curves of both $\bar{D}_{\text{avg}}$ and $\bar{D}_{\text{max}}$ under sufficient capacity. The curves of the proposed min-delay policy and the max-link-rate policy 
highly overlap, which matches the result shown in Fig.~\ref{fig:result-2x1-switch}(a). Their $\bar{D}_{\text{avg}}$ and $\bar{D}_{\text{max}}$ are lower than the backpressure policy: For $\bar{D}_{\text{avg}}$, the backpressure policy induces $5\%$ higher delay on average and a maximum of $12\%$ higher delay over the $500$ tested samples; For $\bar{D}_{\text{max}}$, the backpressure induces $61\%$ higher delay on average and a maximum of $159\%$ higher delay over the $500$ tested samples. 
Fig.~\ref{fig:32x1_limited} illustrates the results under limited capacity. A major contrast to the sufficient capacity case is the significantly poor delay performance of the max-link-rate policy, which echoes Theorem \ref{thm:Nx1} and Fig.~\ref{fig:result-2x1-switch}. We find that the $\bar{D}_{\text{avg}}$ and $\bar{D}_{\text{max}}$ of the max-link-rate policy are $18\%$ and $123\%$ higher than the proposed min-delay policy respectively on average.

\begin{figure}[!htbp]
\centering
\includegraphics[width=1.0\linewidth]{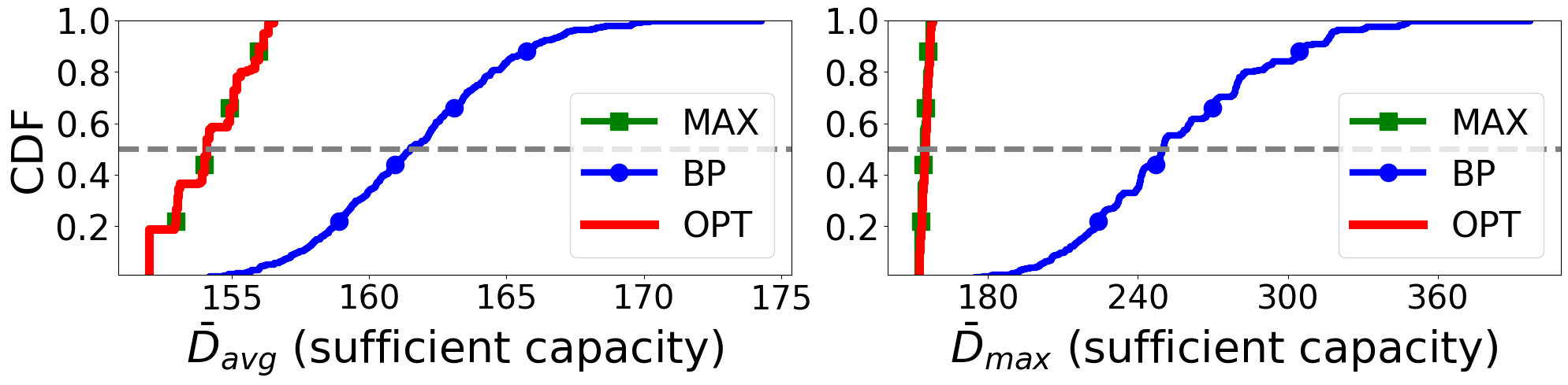}
\caption{CDFs of $\bar{D}_{\text{avg}}$ and $\bar{D}_{\text{max}}$ in $32\times 1$ single-hop networks with sufficient capacity (OPT and MAX are overlapped)}
\label{fig:32x1_suff}
\end{figure}


\begin{figure}[!htbp]
\centering
\includegraphics[width=1.0\linewidth]{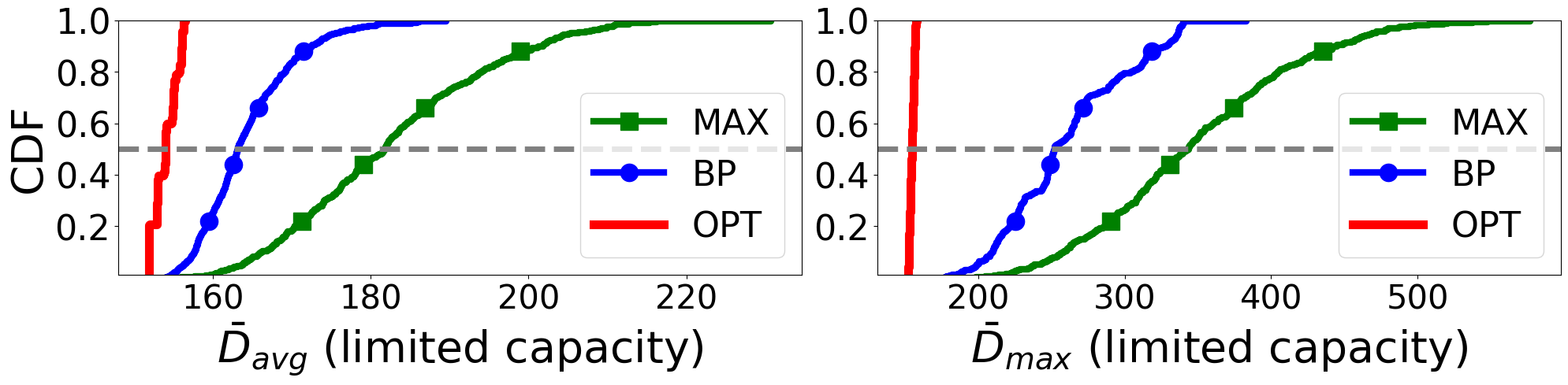}
\caption{CDFs of $\bar{D}_{\text{avg}}$ and $\bar{D}_{\text{max}}$ in $32\times 1$ single-hop networks with limited capacity}
\label{fig:32x1_limited}
\end{figure}


\subsection{General Single-Hop Networks}

We evaluate the delay performance over $32\times 16$ single-hop networks. We consider $500$ different combinations of parameter settings based on the following rules: (i) The arrival rate $\lambda_i$ to each ingress node $s_i$ is uniformly distributed in $[60,100]$; (ii) The service rate $\mu_j$ of each egress node $d_j$ is determined as $0.4\times \alpha_j\times \sum_{i=1}^{32} \lambda_i$ where $\alpha_j$ is a randomly picked weight for $d_j$ with $\sum_{j=1}^{16} \alpha_j = 1$. We consider sufficient link capacity for each pair of ingress and egress nodes. We evaluate the $\bar{D}_{\text{avg}}$ and $\bar{D}_{\text{max}}$ of packets that arrive within the first $50$ time units.


Fig.~\ref{fig:32x16} illustrates the CDF curves of both $\bar{D}_{\text{avg}}$ and $\bar{D}_{\text{max}}$. Table \ref{tab:statistics_performance_gap_single_hop} summarizes the mean and maximum ratio, in terms of both $\bar{D}_{\text{avg}}$ and $\bar{D}_{\text{max}}$, between the backpressure policy, or max-link-rate policy and the min-delay policy. Results show significant reduction of $\bar{D}_{\text{avg}}$ and $\bar{D}_{\text{max}}$ under policies that follow \eqref{eqn:queue_policy_NxM}: On average the backpressure policy incurs $32\%$ higher $\bar{D}_{\text{avg}}$ and $111\%$ higher $\bar{D}_{\text{max}}$, while these metrics for the max-link-rate policy are $258\%$ and $1166\%$ higher than the min-delay policy. We also observe that the worst case for both the backpressure and the max-link-rate policy leads to more than $10\times$ delay compared with the min-delay policies in terms of both $\bar{D}_{\text{avg}}$ and $\bar{D}_{\text{max}}$. We further find that unlike $N\times 1$ networks, the max-link-rate policy no longer minimizes the delay in general single-hop networks in spite of sufficient capacity.
Moreover, we demonstrate that the $\bar{D}_{\text{avg}}$ and $\bar{D}_{\text{max}}$ of the proposed min-delay policy keep stable among all the tested cases. This matches Theorem \ref{thm:optimum_queue_NxM} where the minimum $\bar{D}_{\text{avg}}$ and $\bar{D}_{\text{max}}$ depend only on the ratio  $\sum_{i=1}^{N_S} \lambda_i / \sum_{j=1}^{N_D} \mu_j$ which are the same among all the tested samples\footnote{There are mild fluctuations over different tested samples due to rounding the real numbers to integers.}.

\begin{figure}[!htbp]
\centering
\includegraphics[width=1.0\linewidth]{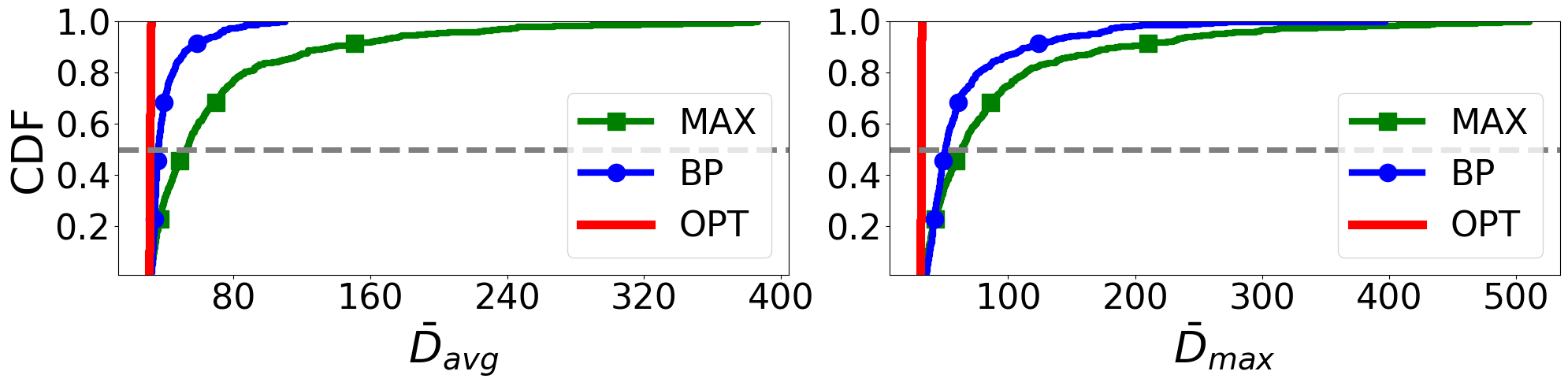}
\caption{CDFs of $\bar{D}_{\text{avg}}$ and $\bar{D}_{\text{max}}$ in $32\times 16$ single-hop networks}
\label{fig:32x16}
\end{figure}

\begin{table}[!htbp]
\centering
\begin{adjustbox}{width=\columnwidth}
\begin{tabularx}{0.95\linewidth}
{c|*{2}{>{\centering\arraybackslash}X}|*{2}{>{\centering\arraybackslash}X}}
    \toprule
     & \multicolumn{2}{c|}{$\bar{D}_{\text{avg}}$} & \multicolumn{2}{c}{$\bar{D}_{\text{max}}$} \\
     & BP/OPT & MAX/OPT & BP/OPT & MAX/OPT \\
    \midrule
    Mean &    1.32 &    2.32 &   2.11 &    2.93 \\
    Max  &    3.58 &   12.38 &  12.66 &   16.29 \\
    \bottomrule
    \end{tabularx}
\end{adjustbox}
\caption{Mean and maximum ratios between the two policies for comparison and the min-delay policy in terms of $\bar{D}_{\text{avg}}$ and $\bar{D}_{\text{max}}$ in $32\times 16$ single-hop networks}
\label{tab:statistics_performance_gap_single_hop}
\end{table}


\subsection{General Multi-Stage Networks}

We further evaluate the performance of the tested policies over general multi-stage networks with $L$ layers. We consider 7 multi-stage network structures listed in Table \ref{table:multi_stage_complete}, including different numbers of layers and fan-in-fan-out topologies. We consider full connection between adjacent layers. We consider $500$ different combinations of parameter settings based on the following rules: (i) The arrival rate $\lambda_i$ to each node $n_i^1$ at the ingress layer is uniformly distributed in $[30,50]$; (ii) The service rate $\mu_j$ of each node $n_j^L$ at the egress layer is determined as $0.4 \times \alpha_j\times \sum_{i=1}^{|\mathcal{V}_1|} \lambda_i$ where $\alpha_j$ is a randomly picked weight for $n_j^L$ with $\sum_{j=1}^{|\mathcal{V}_L|} \alpha_j = 1$. 
We consider sufficient capacity over each pair of nodes at adjacent layers. We evaluate the $\bar{D}_{\text{avg}}$ and $\bar{D}_{\text{max}}$ of packets that arrive within the first $50$ time units.

Table \ref{table:multi_stage_complete} lists the evaluation results over all the tested multi-stage networks.
Columns 1 and 2 show the mean and maximum ratio between the backpressure policy, or the max-link-rate policy and the proposed min-delay policy with respect to $\bar{D}_{\text{avg}}$ and $\bar{D}_{\text{max}}$ respectively. Column 3 shows $\bar{D}_{\text{max}}/\bar{D}_{\text{avg}}$ under the backpressure policy and the max-link-rate policy, which reflects the level of delay fairness of packets injected into different ingress nodes. Note that $\bar{D}_{\text{max}}/\bar{D}_{\text{avg}}=1$ under the min-delay policy as given in Theorem \ref{thm:clos-static-flow}. {The key takeaway is that both the backpressure and max-link-rate policies lead to at least $1.3\times$ $\bar{D}_{\text{avg}}$ and $\bar{D}_{\text{max}}$ of the min-delay policy on average in all $7$ tested structures,}  demonstrating the delay reduction of the proposed policy in general multi-stage networks. Moreover the mean imbalance ratios $\bar{D}_{\text{max}}/\bar{D}_{\text{avg}}$ for both the backpressure and max-link-rate policy are up to $20\%$ higher than the optimal case over the tested instances. We visualize the CDFs of the $\bar{D}_{\text{avg}}$ and $\bar{D}_{\text{max}}$ of two multi-stage networks: $16\times 12 \times 8 \times 6$ in Fig.~\ref{fig:8x6x4x3} and $15\times 12 \times 9 \times 12 \times 15$ in Fig.~\ref{fig:5x4x3x4x5} 
for detailed characterization of their distributions. We can observe the backpressure and the max-link-rate policy have different performance in these two topologies, while our proposed policy achieves minimum delay in all the tested instances.

\begin{table}[!htbp]
\begin{adjustbox}{width=\columnwidth}
\begin{tabular}{c|c|cc|cc|cc}
\toprule
     \multirow{2}{*}{\textbf{Topology}}           &  \multirow{2}{*}{\textbf{Policy}}   & \multicolumn{2}{c|}{Ratio $(\bar{D}_{\text{avg}})$} & \multicolumn{2}{c|}{Ratio $(\bar{D}_{\text{max}})$} & \multicolumn{2}{c}{$\bar{D}_{\text{max}} / \bar{D}_{\text{avg}}$} \\
               &     &               Mean &   Max &               Mean &   Max &                          Mean &   Max \\
\midrule
\multirow{2}{*}{16x12x16} 
               & BP &               1.46 &  2.76 &               1.71 &  3.17 &                          1.17 &  1.63 \\
               & MAX &               1.46 &  2.16 &               1.51 &  2.23 &                          1.03 &  1.11 \\
               \hline
\multirow{2}{*}{12x16x12} 
               & BP &               1.77 &  3.06 &               2.14 &  4.55 &                          1.20 &  1.65 \\
               & MAX &               1.47 &  2.29 &               1.52 &  2.43 &                          1.03 &  1.08 \\
               \hline
\multirow{2}{*}{16x12x8x6} 
               & BP &               1.33 &  1.85 &               1.42 &  2.02 &                          1.07 &  1.21 \\
               & MAX &               1.50 &  3.53 &               1.53 &  3.62 &                          1.02 &  1.07 \\
               \hline
\multirow{2}{*}{6x8x12x16}
               & BP &               1.31 &  2.41 &               1.34 &  2.63 &                          1.02 &  1.11 \\
               & MAX &               1.35 &  2.13 &               1.38 &  2.16 &                          1.02 &  1.06 \\
               \hline
\multirow{2}{*}{15x12x9x12x15} 
               & BP &               1.34 &  1.93 &               1.37 &  2.11 &                          1.02 &  1.15 \\
               & MAX &               1.49 &  2.28 &               1.52 &  2.34 &                          1.02 &  1.07 \\
               \hline
\multirow{2}{*}{9x12x15x12x9} 
               & BP &               1.53 &  2.70 &               1.56 &  2.71 &                          1.02 &  1.09 \\
               & MAX &               1.45 &  2.74 &               1.47 &  2.76 &                          1.01 &  1.07 \\
               \hline
\multirow{2}{*}{12x12x12x12x12} 
               & BP &               1.41 &  2.49 &               1.44 &  2.72 &                          1.02 &  1.09 \\
               & MAX &               1.51 &  2.64 &               1.54 &  2.70 &                          1.02 &  1.07 \\
\bottomrule
\end{tabular}
\end{adjustbox}
\caption{Mean and maximum ratios between each of the two tested policies for comparison and the proposed min-delay policy in terms of $\bar{D}_{\text{avg}}$ (Column 1) and $\bar{D}_{\text{max}}$ (Column 2), and the ratios $\bar{D}_{\text{max}}/\bar{D}_{\text{avg}}$ of the tested policies themselves reflecting the delay imbalance of packets injected into different ingress nodes (Column 3), over $7$ multi-stage topologies.}
\label{table:multi_stage_complete}
\end{table}

\begin{figure}[!htbp]
\centering
\includegraphics[width=1.0\linewidth]{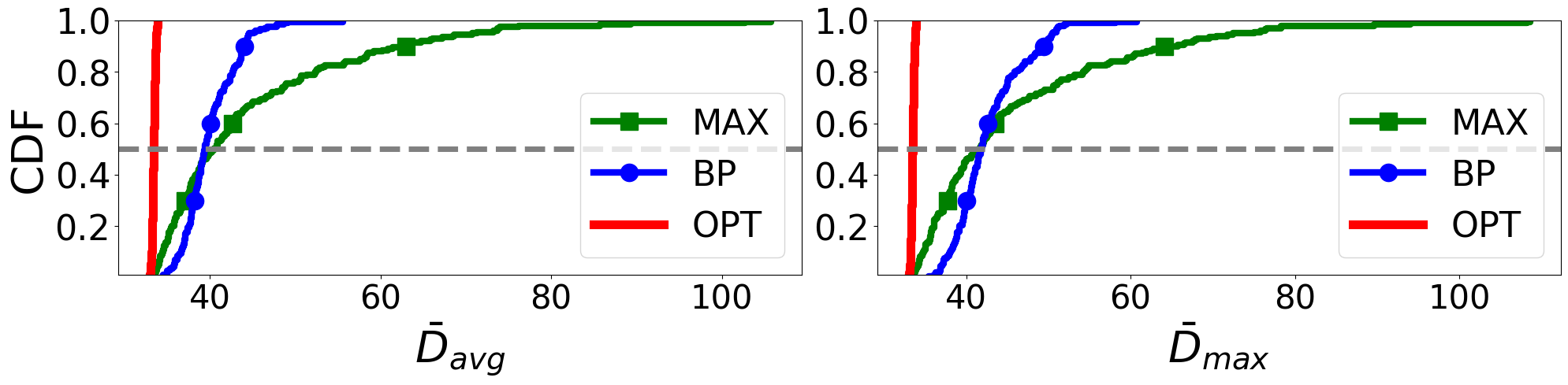}
\caption{CDFs of $\bar{D}_{\text{avg}}$ and $\bar{D}_{\text{max}}$ in $16\times 12 \times 8 \times 6$ networks}
\label{fig:8x6x4x3}
\end{figure}

\begin{figure}[!htbp]
\centering
\includegraphics[width=1.0\linewidth]{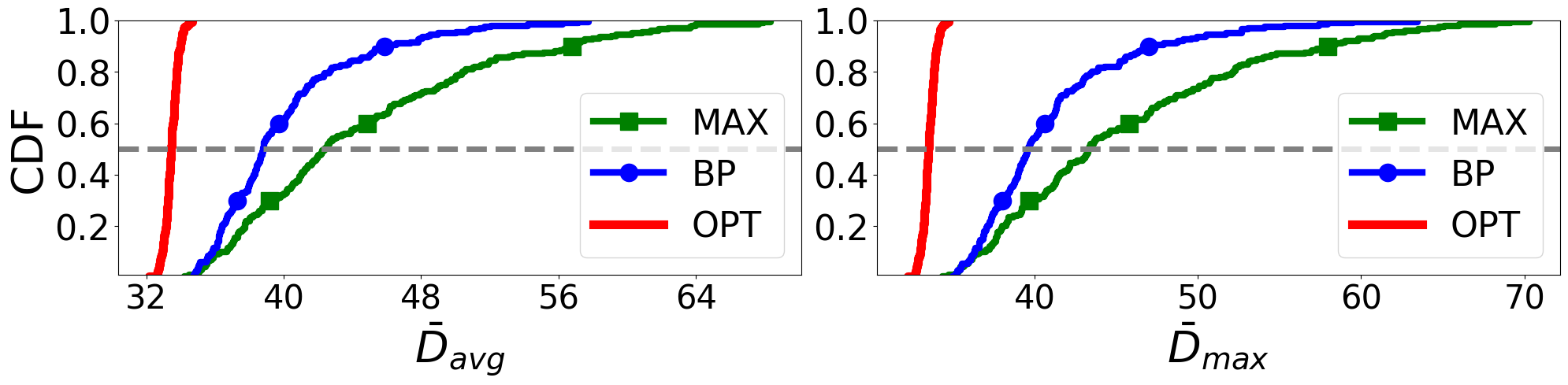}
\caption{CDFs of $\bar{D}_{\text{avg}}$ and $\bar{D}_{\text{max}}$ in $15\times 12 \times 9 \times 12 \times 15$ networks}
\label{fig:5x4x3x4x5}
\end{figure}

\section{Extensions and Discussions}
\label{sec:extension}

In this section, we extend the main results on min-delay policy design to a wider range of scenarios both in practice and in theory. For practice, we demonstrate that the explicit form of the min-delay condition \eqref{eqn:clos-static-flow} (i) facilitates co-optimization of delay and other metrics in data center networks with routing constraints, (ii) provides high flexibility in selecting $\boldsymbol{\lambda}$ based on different objectives in network control, and (iii) guarantees feasibility when adjacent layers are fully connected with sufficient capacity. For theory, we derive a more relaxed min-delay condition for tree data center structures than the general condition \eqref{eqn:clos-static-flow}, and discuss our conjecture on the sufficient and necessary condition for delay minimization in multi-stage networks. We consider static policies in the following discussion for brevity, where the queue-based policies follow the same idea.

\subsection{Practical Extensions}
\label{subsec:practical}

\subsubsection{Multi-Objective Optimization with Routing Constraints}

We have shown that a transmission rate vector $\mathbf{g}$ can achieve global minimum queueing delay if it is a feasible solution subject to the constraints \eqref{eqn:clos-static-flow}. This result implies that we can co-optimize delay with some other metric $f(\mathbf{g})$ under the min-delay constraints \eqref{eqn:clos-static-flow} as follows.

\begin{equation}
\label{eqn:multi-objective}
\min_{\mathbf{g}} \quad f(\mathbf{g}) \qquad 
\textrm{s.t.} \quad
\mathbf{g} \in \eqref{eqn:clos-static-flow}.
\end{equation}

As long as $f(\mathbf{g})$ is a convex function, \eqref{eqn:multi-objective} is a convex optimization problem that can be efficiently solved since the min-delay constraints \eqref{eqn:clos-static-flow} are linear given any $\boldsymbol{\gamma}\in \mathbb{R}_+$. Common examples of $f(\mathbf{g})$ include (i) minimizing total required bandwidth $f(\mathbf{g}) = \sum_{(i,j)\in \mathcal{E}} g_{ij}$ \cite{zhang2005designing}; (ii) minimizing the maximum link rate $f(\mathbf{g}) = \max_{(i,j)\in \mathcal{E}} g_{ij} / c_{ij}$ and average link rate $f(\mathbf{g}) =  |\mathcal{E}|^{-1} \sum_{(i,j)\in \mathcal{E}} \left(g_{ij} / c_{ij} \right)$ \cite{zhang2021gemini}; (iii) minimizing the maximum queue overload rate $f(\mathbf{g}) = \max_{i\in \mathcal{V}} \sum_{k:(k,i)\in \mathcal{E}} g_{ki} - \sum_{j:(i,j)\in \mathcal{E}} g_{ij}$ \cite{wu2022overload}; (iv) minimizing the maximum total queue growth at a layer $f(\mathbf{g}) = \max_{l=1,\dots,L} \sum_{i\in \mathcal{V}_l} \left(\sum_{k:(k,i)\in \mathcal{E}} g_{ki} - \sum_{j:(i,j)\in \mathcal{E}} g_{ij}\right)$. The high generality of $f(\mathbf{g})$ demonstrates the wide application of \eqref{eqn:multi-objective} for multi-objective optimization together with delay minimization in data center networks. 

We can further add constraints on the transmission rate vector $\mathbf{g}$ in \eqref{eqn:multi-objective} in order to capture the restrictions on link rate control. Note that introducing any convex constraint on $\mathbf{g}$ does not violate the convex property of \eqref{eqn:multi-objective}. We give some common examples. (i) Packets cannot go through link $(i,j)$: $g_{ij} = 0$; (ii) A maximum ratio $\beta$ of packets can be transmitted from node $i$ to $j$: $g_{ij} \leq \beta \left(\sum_{k: (k,i)\in \mathcal{E}} g_{ki}\right)$; (iii) The maximum link utilization should not be higher than a threshold $\theta$ in order to avoid link overflow: $\max_{(i,j)\in \mathcal{E}} g_{ij} / c_{ij} \leq \theta$.

\subsubsection{Choices of $\boldsymbol{\gamma}$}

Theorem \ref{thm:clos-static-flow} shows that both of the minimum $\bar{D}_{\text{avg}}$ and $\bar{D}_{\text{max}}$ are independent of $\boldsymbol{\gamma}$, however we show that $\boldsymbol{\gamma}$ affects the overload levels at different nodes and layers, where a lower value of $\gamma_l$ means a higher proportion of packets ingress to layer $l$ will be backlogged. 
We give an example on choosing $\boldsymbol{\gamma}$ to balance the total queue growth rates among all the $L$ layers. Denote the total arrival rates at layer 1 by $\lambda_{\text{sum}}:=\sum_{i=1}^{|\mathcal{V}_1|}\lambda_i$ and the total service rates at layer $L$ by $\mu_{\text{sum}}:=\sum_{j=1}^{|\mathcal{V}_L|} \mu_j$. Then the state we pursue is that the total queue growth rate at each layer is equal to $(\lambda_{\text{sum}} - \mu_{\text{sum}}) / L$. The corresponding $\boldsymbol{\gamma}$ that reaches this state is 
\begin{equation}
\label{eqn:gamma_choice}
\gamma_l = \frac{\lambda_{\text{sum}}-\frac{l-1}{L}(\lambda_{\text{sum}} - \mu_{\text{sum}})}{\lambda_{\text{sum}}-\frac{l}{L}(\lambda_{\text{sum}} - \mu_{\text{sum}})}, ~l = 1,\dots,L.
\end{equation}
Note that we can easily derive a weighted version of \eqref{eqn:gamma_choice} based on the node buffers that can balance the buffer utilization at each layer \cite{ addanki2022abm, roy2015inside, zhang2017high}. The main point we convey from this example is the high flexibility of choosing $\boldsymbol{\gamma}$ according to different objectives in network control.

\subsubsection{Feasibility}

We point out that it is difficult to characterize the conditions on the existence of a feasible transmission rate vector $\mathbf{g}$ subject to \eqref{eqn:clos-static-flow} in multi-stage networks under general network topology and link capacities. Instead, we show in Proposition \ref{prop:existence} the existence of a solution under any given $\boldsymbol{\gamma}\in \mathbb{R}_+$, considering full connection between adjacent layers and sufficient link capacities so that the minimum network cut is $\sum_{i=1}^{|\mathcal{V}_L|} \mu_j$. 
We leave the discussion of general topologies and link capacities to future work.

\begin{proposition}
\label{prop:existence}
Consider an $L$-layer multi-stage network with full connections and sufficient link capacities between adjacent layers. Given the arrival rate vector $\boldsymbol{\lambda}$ and the service rate vector $\boldsymbol{\mu}$, there exists a feasible $\mathbf{g}$ that satisfies the min-delay conditions \eqref{eqn:clos-static-flow} under any given $\boldsymbol{\gamma}\in \mathbb{R}_+$.
\end{proposition}

\begin{proof}
Given a feasible $\boldsymbol{\gamma}$, we can construct a feasible solution to \eqref{eqn:clos-static-flow} as follows based on the full connection and sufficient capacity constraints. For the $l$-th layer $(l < L)$, we first set its total egress rates as $\sum_{n_i^l \in \mathcal{V}_l} \sum_{n_j^{l+1}: (n_i^l,n_j^{l+1})\in\mathcal{E}} g_{n_i^l, n_j^{l+1}}= \left(\sum_{i=1}^{|\mathcal{V}_1|} \lambda_i\right)\prod_{l^{\prime}=1}^{l} \gamma_l$, and we randomly set a positive value for each $n_i^{l} \in \mathcal{V}_l$, denoted by $g_{n_i^l}^{E}$, which represents the egress rate of node $n_i^{l}$, subject to $\sum_{n_i^l\in \mathcal{V}_l} g_{n_i^l}^{E} = \sum_{n_i^l \in \mathcal{V}_l} \sum_{n_j^{l+1}: (n_i^l,n_j^{l+1})\in\mathcal{E}} g_{n_i^l, n_j^{l+1}}$. We can now figure out a feasible solution that satisfies \eqref{eqn:clos-static-flow} by setting 
$$
g_{n_i^l, n_j^{l+1}} = \gamma_l \times  g_{n_i^l}^{I} \times \left(g_{n_j^{l+1}}^{E} / \sum_{k=1}^{|\mathcal{V}_{l+1}|} g_{n_k^{l+1}}^{E}\right)
$$ 
for $\forall n_i^l\in \mathcal{V}_l$ and $\forall n_j^{l+1}\in \mathcal{V}_{l+1}$ iteratively from layer 1 to $L-1$ to calculate $g_{n_i^l}^I$, the ingress to each node $n_i^l$, where we define $g_{n_i^{1}}^{I} := \lambda_i$ for each of the ingress nodes $n_i^{1}$ and $g_{n_i^{l}}^{I} := \sum_{n_k^{l-1}: (n_k^{l-1}, n_i^l) \in \mathcal{E}} g_{n_k^{l-1}, n_i^l}$.
\end{proof}

\textbf{Remark:} The existence of feasible solutions to \eqref{eqn:clos-static-flow} itself does not indicate minimum delay, which requires achieving maximum throughput in the meantime.

\subsection{Theoretical Extensions}
\label{subsec:theory}

\subsubsection{Tree Data Center Structure} A tree structure is a special case of the multi-stage network whose undirected topology is a tree. We visualize an example of a 3-layer tree in Fig.~\ref{fig:Paper_tree_example}, which contains $6$ ingress nodes and a single egress node. We focus on the fan-in structure where $|\mathcal{V}_{l+1}| \leq |\mathcal{V}_l|, \forall l=1,\dots,L-1$ below, while the discussion of general fan-out structures follows in a similar manner.
We derive a sufficient condition on the link rates to achieve minimum $\bar{D}_{\text{avg}}$ and $\bar{D}_{\text{max}}$ in Theorem \ref{thm:fat-tree} based on the \emph{parent source set} of an  node $n_i^l$ at a layer $l$, denoted by $PSS[n_i^l]$, which is defined as all the nodes at the ingress layer that are the parents of $n_i^l$, and we set $PSS[n_i^1] = \{n_i^1\}$ for each ingress node $n_i^1$. In Fig.~\ref{fig:Paper_tree_example}, we have $PSS[n_1^2] = \{n_1^1, n_2^1\}$, $PSS[n_2^2] = \{n_3^1, n_4^1\}$, $PSS[n_3^2] = \{n_5^1, n_6^1\}$, and $PSS[n_1^3] = \mathcal{V}_S$.

\begin{figure}[!htbp]
\centering
\includegraphics[width=0.98\linewidth]{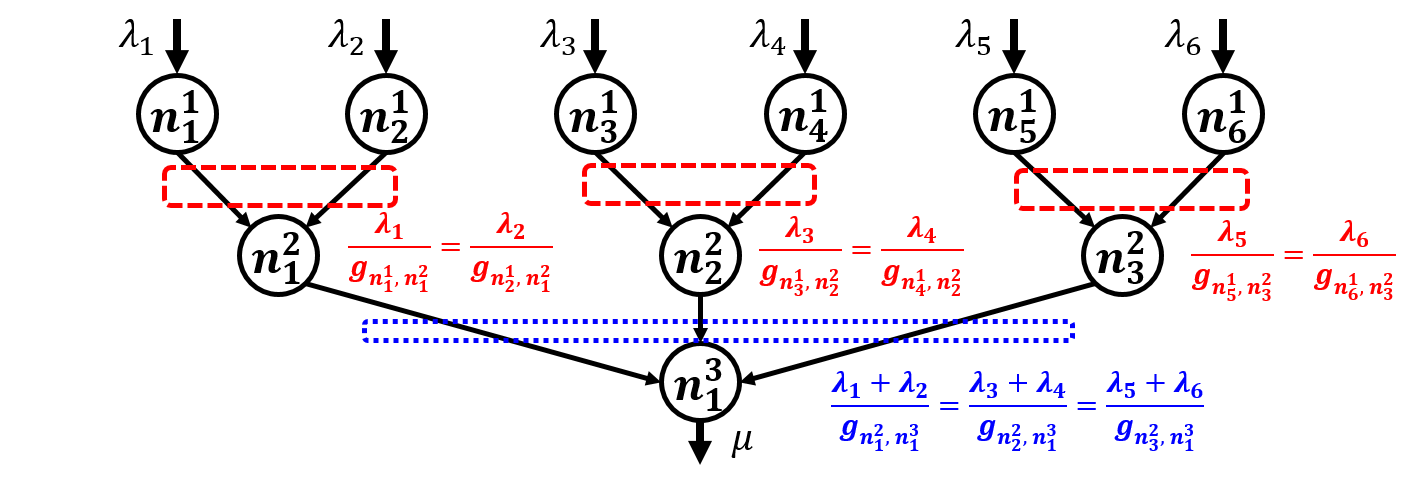}
\caption{A 3-layer fan-in tree data center example under min-delay transmission policies derived in Theorem \ref{thm:fat-tree}.}
\label{fig:Paper_tree_example}
\end{figure}

\begin{theorem}
\label{thm:fat-tree}
Consider an $L$-layer tree data center structure. Any transmission policy $\mathbf{g}$ that guarantees maximum throughput and in the meantime satisfies the following conditions achieves minimum $\bar{D}_{\text{avg}}$ and $\bar{D}_{\text{max}}$: 
For $\forall l \in 2,\dots,L$ and $\forall n_j^{l}\in \mathcal{V}_l$, we have  $\frac{\sum_{k \in PSS[n_i^{l-1}]} \lambda_k}{{g}_{n_i^{l-1}, n_j^l}} = \gamma_{n_j^l}$ for some $\gamma_{n_j^l}>0$ for $\forall n_i^{l-1} \in \mathcal{V}_{l-1}$ with $(n_i^{l-1}, n_j^l)\in \mathcal{E}$.
\end{theorem}

We prove Theorem \ref{thm:fat-tree} under a $3$-layer $4\times 2\times 1$ tree structure in the appendix which can be applied to general multi-layer tree structures. We explain Theorem \ref{thm:fat-tree} through the example in Fig.~\ref{fig:Paper_tree_example}. Layer 2 contains three nodes $n_1^2$, $n_2^2$, $n_3^2$. The parents of these three nodes at layer 1 are $\{n_1^1, n_2^1\}$, $\{n_3^1, n_4^1\}$, and $\{n_5^1, n_6^1\}$ respectively. Then a sufficient condition to achieve minimum delay in Theorem \ref{thm:fat-tree} corresponding to these three nodes at layer 2 are 
\begin{equation}
\label{eqn:fat_tree_1}
\begin{cases}
\lambda_1 / g_{n_1^1,n_1^2} = \lambda_2 / g_{n_2^1, n_1^2} = \gamma_{n_1^2} \\
\lambda_3 / g_{n_3^1,n_2^2} = \lambda_4 / g_{n_4^1, n_2^2} = \gamma_{n_2^2} \\
\lambda_5 / g_{n_5^1,n_3^2} = \lambda_6 / g_{n_6^1, n_3^2} = \gamma_{n_3^2}
\end{cases}
\end{equation}
for some $\gamma_{n_1^2}, \gamma_{n_2^2}, \gamma_{n_3^2}>0$. Similarly, layer 3 contains a single node $n_1^3$, whose parents are $\{n_1^2, n_2^2, n_3^2\}$. The corresponding condition for $n_1^3$ is that for some $\gamma_{n_1^3}>0$,
\begin{equation}
\label{eqn:fat_tree_2}
\frac{\lambda_1 + \lambda_2}{g_{n_1^2,n_1^3}} = \frac{\lambda_3 + \lambda_4}{g_{n_2^2,n_1^3}} = \frac{\lambda_5 + \lambda_6}{g_{n_3^2,n_1^3}} = \gamma_{n_1^3}
\end{equation}
Combining \eqref{eqn:fat_tree_1} and \eqref{eqn:fat_tree_2} gives a sufficient condition on $\mathbf{g}$ to minimize $\bar{D}_{\text{avg}}$ and $\bar{D}_{\text{max}}$ for the 3-layer tree in Fig.~\ref{fig:Paper_tree_example} as long as maximum throughput is achieved. 

An implication of Theorem \ref{thm:fat-tree} is that its min-delay condition also follows the rate-proportional property, where the transmission rates of the egress links from the same layer should be in the same proportion to the total external packet arrival rates among their parent source sets. Moreover, it is straightforward to verify that any policy that satisfies the general min-delay condition \eqref{eqn:clos-static-flow} in the tree structure must satisfy the condition in Theorem \ref{thm:fat-tree}. This means Theorem \ref{thm:fat-tree} gives a more relaxed sufficient condition on link rate control to minimize delay than applying \eqref{eqn:clos-static-flow} directly to the tree.





\subsubsection{Conjecture on Sufficient and Necessary Conditions}
\label{subsec:conjecture_necessary}

We have yet derived the necessary conditions to achieve minimum delay in general multi-stage networks. The challenge is to characterize all the min-delay policies due to the possible switching from non-zero to zero queue length at some nodes whose egress rates are greater than  ingress rates in the time window $[t_0, t_0+T]$. The $N\times 1$ network is the special case that we can fully characterize the complete set of min-delay policies, as shown in Theorem \ref{thm:Nx1}. 

We conjecture that the crux lies in that the actual transmission rate of a link $(i,j)$ starting from a node with zero queue length are not equal to the transmission rate we originally set.  Consider the example where a node $i$ at a layer $l$ with links $(1,i)$, $(2,i)$, $(i,3)$, and $(i,4)$, where node 1 and 2 are at layer $l-1$, node 3 and 4 are at layer $l+1$. Suppose we set $g_{1i} = 1$, $g_{2i} = 2$, $g_{i3} = g_{i4} = 3$. Clearly $g_{1i} + g_{2i} < g_{i3} + g_{i4}$ thus $q_i(t) = 0$ starting from some time. Since the definition of link rate is the number of packets transmitted over this link in a time unit, there are $3$ units of packets injected into node $i$ at time $t$, and $1.5$ units will be served to each of the egress links $(i,3)$ and $(i,4)$, and thus the actual transmission rates over $(i,3)$ and $(i,4)$ are both $1.5$, less than $g_{i3}=g_{i4}=3$. Following this idea, given the transmission rate vector $\mathbf{g}$ we set, we can obtain the actual transmission rates denoted by $\tilde{\mathbf{g}}$ based on \eqref{eqn:actual_service_rates} starting from the ingress layer in multi-stage networks: For $\forall l = 1,\dots, L-1, ~\forall i \in \mathcal{V}_l, j \in \mathcal{V}_{l+1}$,
\begin{equation}
\label{eqn:actual_service_rates}
\small
\begin{aligned}
\tilde{g}_{ij} &= 
\begin{cases}
g_{ij}, \text{if} \sum\limits_{k\in\mathcal{V}_{l-1}} \tilde{g}_{ki} \geq  \sum\limits_{k\in\mathcal{V}_{l+1}} g_{ik} \\
g_{ij} \frac{\sum_{k\in\mathcal{V}_{l-1}} \tilde{g}_{ki}}{\sum_{k\in\mathcal{V}_{l+1}} g_{ik}}, \text{o.w.}
\end{cases} \hspace{-5mm}= g_{ij}\min\left\{1,\frac{\sum_{k\in\mathcal{V}_{l-1}} \tilde{g}_{ki}}{\sum_{k\in\mathcal{V}_{l+1}} g_{ik}}\right\} 
\end{aligned}
\end{equation}
where we define $\sum_{k\in \mathcal{V}_0} \tilde{g}_{ki}:=\lambda_i, ~\forall i \in \mathcal{V}_1$ to incorporate the ingress layer in \eqref{eqn:actual_service_rates}. We have the following conjecture of the sufficient and necessary condition to minimize $\bar{D}_{\text{avg}}$ and $\bar{D}_{\text{max}}$ based on actual transmission rate vector $\tilde{\mathbf{g}}$.




\begin{conjecture}
\label{conj:sufficient_and_necessary}
Consider an $L$-layer multi-stage network. The sufficient and necessary condition of a transmission policy $\mathbf{g}$ to minimize both $\bar{D}_{\text{avg}}$ and $\bar{D}_{\text{max}}$ is that its corresponding actual transmission rate vector $\tilde{\mathbf{g}}$ satisfies \eqref{eqn:clos-static-flow} and achieves maximum throughput.
\end{conjecture}

We can validate the conjecture in $N\times 1$ networks based on Fig.~\ref{fig:Nx1}: Any $\mathbf{g}$ such that $g_i\geq \lambda_i,~\forall i =1,\dots,N$ leads to minimum delay, and its corresponding $\tilde{g}_i = \lambda_i$ and thus $\{\tilde{g}_i\}_{i=1}^N$ are in proportion to $\{\lambda_i\}_{i=1}^N$. However, if there is a single $i_0\in \mathcal{V}_S$ such that $g_{i_0}<\lambda_{i_0}$, which does not induce minimum delay, then clearly $\tilde{g}_{i_0}=g_{i_0}$ and thus $\tilde{g}_{i_0}/\lambda_{i_0} < \tilde{g}_{i}/\lambda_{i} = 1,~\forall i \neq i_0$, i.e., \eqref{eqn:clos-static-flow} does not hold for $\tilde{\mathbf{g}}$. We further visualize an example of $2\times 2$ networks in Fig.~\ref{fig:Paper_conjecture_2x2_example} to explain the conjecture. Furthermore, we point out that the conjecture can be generalized to queue-based policies, where we require that the corresponding actual transmission rate vector $\tilde{\mathbf{g}}(\mathbf{q}(t))$ meets the queue-proportional condition \eqref{eqn:clos-queue-based} and maximum throughput requirement in Theorem \ref{thm:clos-queue-based-policy}. We leave the proof of this conjecture to future work.


\begin{figure}[!htbp]
\centering
\includegraphics[width=0.98\linewidth]{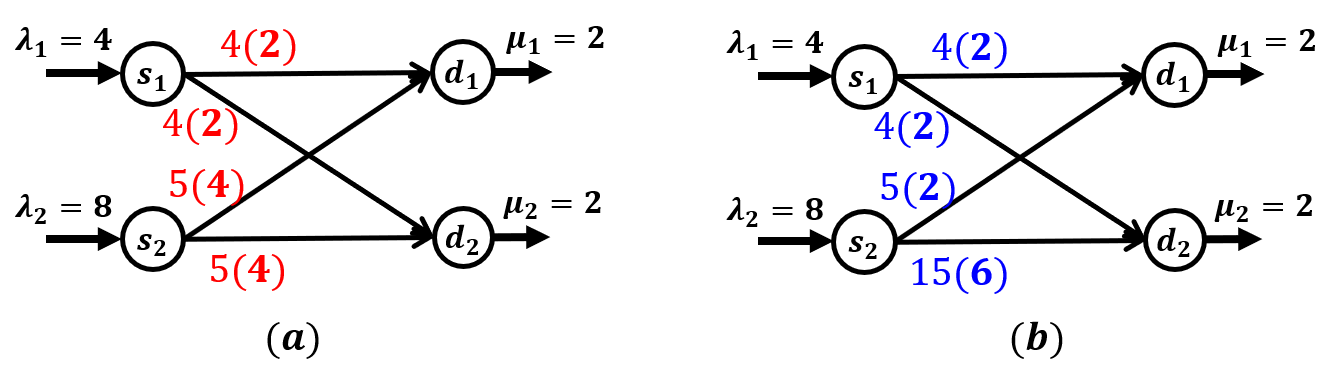}
\caption{Validation of Conjecture \ref{conj:sufficient_and_necessary} in $2\times 2$ single-hop networks, where the notation $x(y)$ over a link $(i,j)$ means $g_{ij}=x$ and its actual transmission rate $\bar{g}_{ij}=y$ calculated based on \eqref{eqn:actual_service_rates}: (a) setting $\mathbf{g}=[4,4,5,5]$ leads to $\tilde{\mathbf{g}}=[2,2,4,4]$ which satisfies \eqref{eqn:clos-queue-based} thus min-delay; (b) setting $\mathbf{g}=[4,4,5,15]$ leads to $\tilde{\mathbf{g}}=[2,2,2,6]$ which does not satisfy \eqref{eqn:clos-queue-based} thus not min-delay.}
\label{fig:Paper_conjecture_2x2_example}
\end{figure}


\section{Conclusion and Future Work}

We study link rate control for queueing delay minimization in overloaded networks. Leveraging the fluid queueing model, we show that any static rate-proportional policy, which guarantees identical ratios between the ingress and egress rates of all the nodes at the same layer, minimizes the average delay $\bar{D}_{\text{avg}}$ and the maximum ingress delay $\bar{D}_{\text{max}}$ in general single-hop and multi-stage networks. We further extend the result to the queue-proportional policies which can achieve asymptotically minimum delay based on real-time queue information agnostic of packet arrival rates. We evaluate the performance of our proposed policies under numerous network settings, validate their min-delay property, and demonstrate their superiority in delay reduction compared with  the backpressure policy and the max-link-rate policy. We finally discuss the extensions of the main results in practice and in theory. We envision multiple future directions including proving sufficient and necessary conditions on link rate control for delay minimization in general multi-stage networks, the extension of main results to multi-hop networks, and the implementation of the proposed policies in real data center networks.

\bibliography{DelayMinimization}
\bibliographystyle{IEEEtran}

\appendix

\subsection{Proof of Theorem \ref{thm:NxM}}

\begin{proof}
We present the proof sketch here. Due to the space limit, we only prove for $\bar{D}_{\text{avg}}$, where $\bar{D}_{\text{max}}$ is similar. For packets that arrive at the ingress node $s_1$ at time $t$ and finally transmitted to the egress node $d_1$, the total queueing delay is
$$
\small
\begin{aligned}
&D_{s_1d_1}(t)=\frac{q_{s_1}(t)}{g_{11}+g_{12}}+\frac{q_{d_1}\left(t+\frac{q_{s_1}(t)}{g_{11}+g_{12}}\right)}{\mu_1}
\\&=\frac{q_{s_1}(t)}{g_{11}+g_{12}} + \frac{1}{\mu_1} \max\left\{0, q_{d_1}(t)+\frac{q_{s_1}(t)}{g_{11}+g_{12}}(g_{11}+g_{21}-\mu)\right\}
\\&=\begin{cases}
\frac{q_{d_1}(t)}{\mu_1} + \frac{q_{s_1}(t)}{\mu_1}\frac{g_{11}+g_{21}}{g_{11}+g_{12}},\qquad g_{11}+g_{21}\geq \mu_1 \\
\frac{q_{s_1}(t)}{g_{11}+g_{12}},\qquad g_{11}+g_{21}<\mu_1
\end{cases}
\end{aligned}
$$
Since $\mathbf{q}(t_0)=\boldsymbol{0}$, then the average delay for packets from $s_1$ to $d_1$, denoted as $\bar{D}_{s_1d_1}$, is
$$
\small
\begin{aligned}
&\bar{D}_{s_1d_1}:=\frac{1}{T}\int_{t_0}^{t_0+T} D_{s_1d_1}(t)dt
\\&=\begin{cases}
\frac{T}{2\mu_1}(g_{11}+g_{21}-\mu_1) \\
\quad + \frac{T}{2\mu_1}\frac{g_{11}+g_{21}}{g_{11}+g_{12}}\max\{\lambda_1-g_{11}-g_{12}, 0\},~g_{11}+g_{21}\geq \mu_1 \\
\frac{T}{2(g_{11}+g_{21})}\max\left\{\lambda_1-g_{11}-g_{12}, 0\right\},~g_{11}+g_{21}\leq \mu_1
\end{cases}
\end{aligned}
$$
We can verify that among all transmission rate vectors $\mathbf{g}$'s that $g_{11}+g_{21}\leq \mu_1$, the $\mathbf{g}$'s that satisfy $g_{11}+g_{21}=\mu_1$ achieve minimum delay\footnote{The intuition is clear that $g_{11}+g_{21}<\mu_1$ does not fully utilize the service capability of node $d_1$.}. Therefore the minimum delay achieved under $g_{11}+g_{21}\geq \mu_1$ is exactly the global optimum, under which
$$
\begin{aligned}
\bar{D}_{s_1d_1}
&=\begin{cases}
\frac{T}{2\mu_1} \left(\lambda_1\frac{g_{11}+g_{21}}{g_{11}+g_{12}}-\mu_1\right), ~g_{11}+g_{12}\leq \lambda_1 \\
\frac{T}{2\mu_1} \left(g_{11}+g_{21}-\mu_1\right),~g_{11}+g_{12}\geq \lambda_1
\end{cases}
\end{aligned}
$$

Generally, we can obtain that for any $(i,j)\in \{(1,1), (1,2), (2,1), (2,2)\}$,
$$
\bar{D}_{s_id_j}
=\begin{cases}
\frac{T}{2\mu_j} \left(\lambda_i\frac{g_{1j}+g_{2j}}{g_{i1}+g_{i2}}-\mu_j\right), ~g_{i1}+g_{i2}\leq \lambda_i \\
\frac{T}{2\mu_j} \left(g_{1j}+g_{2j}-\mu_2\right),~g_{i1}+g_{i2}\geq \lambda_i
\end{cases}
$$

Therefore, we have
$$
\small
\begin{aligned}
\bar{D}_{\text{avg}} &= \frac{\lambda_1}{\lambda_1+\lambda_2} \frac{g_{11}}{g_{11}+g_{12}}\bar{D}_{s_1d_1} + \frac{\lambda_1
}{\lambda_1+\lambda_2} \frac{g_{12}}{g_{11}+g_{12}}\bar{D}_{s_1d_2}\\&
+\frac{\lambda_2}{\lambda_1+\lambda_2} \frac{g_{21}}{g_{21}+g_{22}}\bar{D}_{s_2d_1} + \frac{\lambda_2
}{\lambda_1+\lambda_2} \frac{g_{22}}{g_{21}+g_{22}}\bar{D}_{s_2d_2}
\end{aligned}
$$
where $\frac{\lambda_i}{\lambda_1+\lambda_2}\frac{g_{ij}}{g_{i1}+g_{i2}}$ denotes the portion of packets from $s_i$ to $d_j$ that arrive within $[t_0,t_0+T]$. We again consider the four regions $\{\mathbf{g}\mid g_{11}+g_{12} \lessgtr \lambda_1, g_{21}+g_{22} \lessgtr \lambda_2\}$ and prove that the optimal solutions constrained in each region are all global optimum. Due to space limit we only show the details of the case $g_{11}+g_{12}\leq \lambda_1,~g_{21}+g_{22}\leq \lambda_2$.
%
In this case, the average delay of packets that arrive in $[t_0,t_0+T]$ among all ingress nodes is
$$
\footnotesize
\begin{aligned}
&\bar{D}_{\text{avg}}
\sim \frac{g_{11}+g_{21}}{\mu_1} \left(\frac{\lambda_1^2}{g_{11}}\left(\frac{g_{11}}{g_{11}+g_{12}}\right)^2+\frac{\lambda_2^2}{g_{21}}\left(\frac{g_{21}}{g_{21}+g_{22}}\right)^2\right)
\\&\qquad \quad +\frac{g_{12}+g_{22}}{\mu_2} \left(\frac{\lambda_1^2}{g_{12}}\left(\frac{g_{12}}{g_{11}+g_{12}}\right)^2+\frac{\lambda_2^2}{g_{22}}\left(\frac{g_{22}}{g_{21}+g_{22}}\right)^2\right)
\\&=\frac{1}{\mu_1} \left(\lambda_1^2x^2+\lambda_2^2y^2+\lambda_1^2x^2\frac{g_{21}}{g_{11}}+\lambda_2^2 y^2 \frac{g_{11}}{g_{21}}\right)
\\&\quad +\frac{1}{\mu_2} \left(\lambda_1^2(1-x)^2+\lambda_2^2(1-y)^2+\lambda_1^2(1-x)^2\frac{g_{22}}{g_{12}}+\lambda_2^2 (1-y)^2 \frac{g_{12}}{g_{22}}\right)
\\&
\overset{\text{(i)}}{\geq} \frac{1}{\mu_1}(\lambda_1x+\lambda_2y)^2+\frac{1}{\mu_2}(\lambda_1(1-x)+\lambda_2(1-y))^2 
\\&
\overset{\text{(ii)}}{\geq} \frac{T}{2(\mu_1+\mu_2)}(\lambda_1+\lambda_2) - \frac{T}{2}
= \frac{T}{2(\mu_1+\mu_2)}(\lambda_1+\lambda_2-\mu_1-\mu_2)
\end{aligned}
$$
where $\sim$ means removing constant terms, $x:=\frac{g_{11}}{g_{11}+g_{12}}$, and $y:=\frac{g_{21}}{g_{21}+g_{22}}$. The inequality (i) stems from Cauchy-Schwartz Inequality, which turns into equality when
$\frac{g_{11}}{g_{21}}=\frac{\lambda_1x}{\lambda_2y},~\frac{g_{12}}{g_{22}}=\frac{\lambda_1(1-x)}{\lambda_2(1-y)}$
and equivalently,
\begin{equation}
\label{eqn:2x2-cond1}
\frac{g_{11}+g_{12}}{g_{21}+g_{22}} = \frac{\lambda_1}{\lambda_2}.
\end{equation}
The inequality (ii) holds due to solving
$$
\min_{x,y\in[0,1]} \quad \frac{1}{\mu_1}(\lambda_1x+\lambda_2y)^2+\frac{1}{\mu_2}(\lambda_1(1-x)+\lambda_2(1-y))^2
$$
where the optimal $(x,y)$ satisfies
$$
\begin{cases}
\lambda_1x+\lambda_2y=\frac{\lambda_1+\lambda_2}{\mu_2}\left(\frac{1}{\mu_1}+\frac{1}{\mu_2}\right)^{-1} \\
\lambda_1(1-x)+\lambda_2(1-y)=\frac{\lambda_1+\lambda_2}{\mu_1}\left(\frac{1}{\mu_1}+\frac{1}{\mu_2}\right)^{-1}
\end{cases}
$$
which, combined with \eqref{eqn:2x2-cond1}, is equivalent to 
$$
\begin{cases}
\frac{\lambda_2}{\lambda_1+\lambda_2} \frac{g_{11}+g_{21}}{g_{21}+g_{22}} = \frac{\mu_1}{\mu_1+\mu_2} \\
\\
\frac{\lambda_2}{\lambda_1+\lambda_2} \frac{g_{12}+g_{22}}{g_{21}+g_{22}} = \frac{\mu_2}{\mu_1+\mu_2}
\end{cases}
$$
and thus
\begin{equation}
\label{eqn:2x2-cond2}
\frac{g_{11}+g_{21}}{g_{12}+g_{22}} = \frac{\mu_1}{\mu_2}
\end{equation}
suffices to make the inequality (ii) achieve its lower bound. Therefore \eqref{eqn:2x2-cond1} and \eqref{eqn:2x2-cond2} with $g_{11}+g_{21}\geq \mu_1,~g_{12}+g_{22}\geq \mu_2$ give us the sufficient and necessary condition on $\mathbf{g}$ to minimize $\bar{D}_{\text{avg}}$ under $g_{11}+g_{12}\leq \lambda_1$ and $g_{21}+g_{22}\leq \lambda_2$. 
\end{proof}

\subsection{Proof of Theorem \ref{thm:clos-static-flow}}

We prove the min-delay conditions on static transmission policies under a 3-layer network with 2 nodes at each layer, and links connecting each pair of nodes at adjacent layers, as shown in Fig.~\ref{fig:2x2x2}. Note that we re-index each node from 1 to 6 to increase readability of the proof. The results can be smoothly extended to general $L$-layer networks with arbitrary numbers of nodes at each layer and links between adjacent layers. For simplicity, we consider zero initial queue length of each node.


\begin{figure}[!htbp]
\centering
\includegraphics[width=0.8\linewidth]{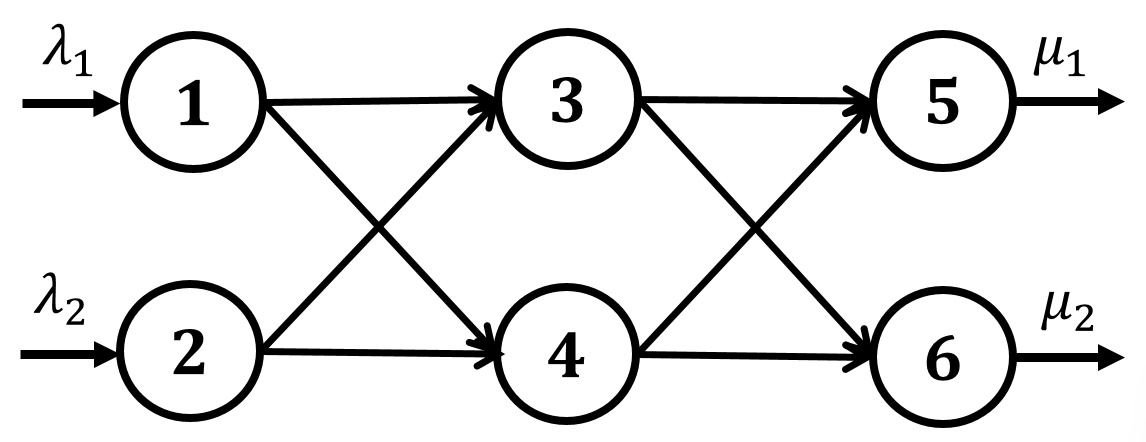}
\caption{An example of a 3-layer $2\times2\times2$ network}
\label{fig:2x2x2}
\end{figure}

Note that we can classify all the packets going through the 3-layer network by their paths. We characterize the total queueing delay of a packet taking the path $1 \rightarrow 3 \rightarrow 5$ as in \eqref{eqn:delay_135}, where the other possible paths are in similar form. We solely present the derivation assuming that the total ingress rate is not lower than the total egress rate at each node for simplification, while the min-delay policy derived in this case holds its optimality in other cases. Recall that $D_p(t)$ represents the queueing delay of a packet arriving to the network at the ingress node $p[0]$ at time t. We have

\begin{equation}
\footnotesize
\begin{aligned}
D_{135}(t) &= \frac{q_1(t)}{g_{13} + g_{14}} + \frac{q_3\left(t + \frac{q_1(t)}{g_{13} + g_{14}}\right)}{g_{35} + g_{36}} \\& \qquad+ \frac{q_5\left(t + \frac{q_1(t)}{g_{13} + g_{14}} + \frac{q_3\left(t + \frac{q_1(t)}{g_{13} + g_{14}}\right)}{g_{35} + g_{36}} \right)}{\mu} 
\\& \overset{(a)}{=} q_1(t) \times \frac{1}{g_{13} + g_{14}} \times\frac{g_{13} + g_{23}}{g_{35} + g_{36}} \times \frac{g_{35} + g_{45}}{\mu_1}
\\& \qquad + q_3(t) \times \frac{1}{g_{35} + g_{36}} \times \frac{g_{35} + g_{45}}{\mu_1} + q_5(t) \times \frac{1}{\mu_1} 
\\& \overset{(b)}{=} \left(\frac{\lambda_1}{g_{13} + g_{14}} \times\frac{g_{13} + g_{23}}{g_{35} + g_{36}} \times\frac{g_{35} + g_{45}}{\mu_1} - 1\right)t
\end{aligned}
\label{eqn:delay_135}
\end{equation}
where $(a)$ stems from the expansion of $q_3\left(t + \frac{q_1(t)}{g_{13} + g_{14}}\right)$ and $q_5\left(t + \frac{q_1(t)}{g_{13} + g_{14}} + \frac{q_3\left(t + \frac{q_1(t)}{g_{13} + g_{14}}\right)}{g_{35} + g_{36}} \right)$ into a linear combination of $q_1(t)$, $q_3(t)$, and $q_5(t)$, for example 
$
q_3\left(t + \frac{q_1(t)}{g_{13} + g_{14}}\right) = q_3(t) + \frac{q_1(t)}{g_{13} + g_{14}} \times (g_{13} + g_{23} - g_{34} - g_{35})
$
and similarly for $q_5\left(t + \frac{q_1(t)}{g_{13} + g_{14}} + \frac{q_3\left(t + \frac{q_1(t)}{g_{13} + g_{14}}\right)}{g_{35} + g_{36}} \right)$; (b) is due to $q_1(t) = (\lambda_1 - g_{13} - g_{14})t$, $q_3(t) = (g_{13} + g_{23} - g_{34} - g_{35})t$, and $q_5(t) = (g_{35} + g_{45} - \mu_1)t$. We can similarly express the delay of packets taking the other 7 paths, and further obtain the time-average values $\bar{D}_p$ for each path $p$ in $[t_0, t_0+T]$.

We can then express the average queueing delay of packets that arrive to the network in overload within time window $[0,T]$ as $\bar{D}_{\text{avg}}:=\sum_{p\in \mathcal{P}} w_p \bar{D}_p$, where $\mathcal{P}$ denotes the set of possible paths of packets, and $w_p$ denotes the proportion of the packets that take the path $p$. For example, 
$
w_{135} = \lambda_1 \frac{g_{13}}{g_{13} + g_{14}}\frac{g_{35}}{g_{35} + g_{36}},
$
and similarly for the other paths. In the following, we derive the conditions on the transmission rate vector $\mathbf{g}$ that minimizes $\bar{D}_{\text{avg}}$. We first consider the sum of the weighted delay for two paths $1\rightarrow 3 \rightarrow 5$ and $2\rightarrow 3 \rightarrow 5$, which can be derived as
$$
\footnotesize
\begin{aligned}
&w_{135}\bar{D}_{135} + w_{235}\bar{D}_{235} \\
&= \frac{g_{13} + g_{23}}{\mu_1} \times \frac{g_{35}(g_{35} + g_{45})}{(g_{35} + g_{36})^2}\times \left(\lambda_1^2 \frac{g_{13}}{(g_{13} + g_{14})^2} + \lambda_2^2 \frac{g_{23}}{(g_{23} + g_{24})^2} \right)
\\&= \frac{1}{\mu_1}\times \frac{g_{35}(g_{35} + g_{45})}{(g_{35} + g_{36})^2} \times \left(\lambda_1^2x^2 + \lambda_2^2y^2 + \lambda_1^2 x^2 \frac{g_{23}}{g_{13}} + \lambda_2^2y^2 \frac{g_{13}}{g_{23}}\right)
\\& \geq \frac{1}{\mu_1} \times \frac{g_{35}(g_{35} + g_{45})}{(g_{35} + g_{36})^2} \left(\lambda_1x + \lambda_2y\right)^2
\end{aligned}
$$
where $x=\frac{g_{13}}{g_{13}+g_{14}}$ and $y = \frac{g_{23}}{g_{23}+g_{24}}$, and the last inequality is based on Cauchy-Schwartz inequality, where the condition of link rates to reach the lower bound is 
\begin{equation}
\label{eqn:equality_condition_2x2x2_1}
\frac{\lambda_1x}{\lambda_2y} = \frac{g_{13} }{g_{23}}.
\end{equation}
We can similarly derive the same condition \eqref{eqn:equality_condition_2x2x2_1} for the sum $w_{136}\bar{D}_{136} + w_{236}\bar{D}_{236}$ reaching its lower bound. For the other two sums $w_{145}\bar{D}_{145} + w_{245}\bar{D}_{245}$ and $w_{146}\bar{D}_{146} + w_{246}\bar{D}_{246}$, the condition becomes 
\begin{equation}
\label{eqn:equality_condition_2x2x2_2}
\frac{\lambda_1(1-x)}{\lambda_2(1-y)} = \frac{g_{14}}{g_{24}}.
\end{equation}
Note that with $x=\frac{g_{13}}{g_{13}+g_{14}}$ and $y = \frac{g_{23}}{g_{23}+g_{24}}$, both \eqref{eqn:equality_condition_2x2x2_1} and \eqref{eqn:equality_condition_2x2x2_2} lead to the condition 
\begin{equation}
\label{eqn:equality_condition_2x2x2}
\frac{g_{13}+g_{14}}{g_{23}+g_{24}} = \frac{\lambda_1}{\lambda_2}.
\end{equation}

Suppose that \eqref{eqn:equality_condition_2x2x2} is satisfied below. We can further extend the above idea to derive the min-delay transmission rates between layer 2 and 3. Specifically, let $z = \frac{g_{35}}{g_{35}+g_{36}}$, $w=\frac{g_{45}}{g_{45} + g_{46}}$, then we have
$$
\begin{aligned}
&w_{135}\bar{D}_{135} + w_{235}\bar{D}_{235} + w_{145}\bar{D}_{145} + w_{245}\bar{D}_{245} 
\\& = \frac{1}{\mu_1}
\big((\lambda_1x + \lambda_2y)^2z^2\big(1+\frac{g_{45}}{g_{35}}\big) 
\\& \quad + (\lambda_1(1-x) + \lambda_2(1-y))^2w^2\big(1+\frac{g_{35}}{g_{45}}\big)\big)
\\& \geq \frac{1}{\mu_1}\left((\lambda_1x + \lambda_2y)z + (\lambda_1(1-x) + \lambda_2(1-y))w\right)^2
\end{aligned}
$$
where the condition that achieves the lower bound is 
\begin{equation}
\label{eqn:equality_condition_2x2x2_stage2}
\begin{aligned}
&\frac{z(\lambda_1x + \lambda_2y)}{w(\lambda_1(1-x) +\lambda_2(1-y))} = \frac{g_{35}}{g_{45}} 
\\& \quad = 
\frac{g_{35} + g_{36}}{g_{45} + g_{46}} = \frac{\lambda_1x + \lambda_2y}{\lambda_1(1-x) + \lambda_2(1-y)}.
\end{aligned}
\end{equation}
The other half $w_{136}\bar{D}_{136} + w_{236}\bar{D}_{236} + w_{146}\bar{D}_{146} + w_{246}\bar{D}_{246}$ leads to the same condition as in \eqref{eqn:equality_condition_2x2x2_stage2}.

Suppose that \eqref{eqn:equality_condition_2x2x2_stage2} is satisfied below. The problem to minimize the average queueing delay now becomes
$$
\small
\begin{aligned}
&\min_{x,y,z,w} \frac{1}{\mu_1}\left((\lambda_1x + \lambda_2y)z + (\lambda_1(1-x) + \lambda_2(1-y))w\right)^2 
\\&  + \frac{1}{\mu_2}\left((\lambda_1x + \lambda_2y)(1-z) + (\lambda_1(1-x) + \lambda_2(1-y))(1-w)\right)^2 
\end{aligned}
$$
Based on the Cauchy-Schwartz inequality, the optimal condition is that 
\begin{equation}
\small
\label{eqn:equality_condition_2x2x2_stage3}
\frac{(\lambda_1x + \lambda_2y)z + (\lambda_1(1-x) + \lambda_2(1-y))w}{(\lambda_1x + \lambda_2y)(1-z) + (\lambda_1(1-x) + \lambda_2(1-y))(1-w)} = \frac{\mu_1}{\mu_2}
\end{equation}

We combine the three optimal conditions \eqref{eqn:equality_condition_2x2x2}, \eqref{eqn:equality_condition_2x2x2_stage2}, and \eqref{eqn:equality_condition_2x2x2_stage3} which altogether reach the lower bound of the average queuing delay of packets. We simplify them and obtain the min-delay constraints that reach the lower bound as follows.
\begin{equation}
\label{eqn:delay_optimal_condition_2x2x2}
\begin{cases}
\frac{\lambda_1}{\lambda_2} = \frac{g_{13} + g_{14}}{g_{23} + g_{24}} \\
\frac{g_{13} + g_{23}}{g_{14} + g_{24}} = \frac{g_{35} + g_{36}}{g_{45} + g_{46}} \\
\frac{g_{35} + g_{45}}{g_{36} + g_{46}} = \frac{\mu_1}{\mu_2}
\end{cases}
\end{equation}
which means to guarantee an identical ratio between the ingress and egress rates of each node at a layer. Combining \eqref{eqn:delay_optimal_condition_2x2x2} and the condition to achieve maximum throughput in the network, we obtain a sufficient condition on a policy that minimizes the average queueing delay. We can derive similarly for $\bar{D}_{\text{max}}$ minimization, which is the maximum average queueing delay of a packet in the network over different ingress nodes. 







\subsection{Proof of Theorem \ref{thm:fat-tree}}

We give the proof of a $4\times 2 \times 1$ tree structure, where $\mathcal{V}_1 = \{n_1^1, n_2^1, n_3^1, n_4^1\}$, $\mathcal{V}_2 = \{n_1^2, n_2^2\}$, $\mathcal{V}_3 = \{n_1^3\}$. We denote the total queueing delay of a packet injected into node $n_i^1$ at the ingress layer at time $t$ as $D_{i}(t)$, which is composed of the queueing delay at each of the 3 layers. For example, 
$$
\begin{aligned}
D_1(t) &= \frac{q_{n_1^1}(t)}{g_{n_1^1, n_1^2}} + \frac{q_{n_1^2}\left(t + \frac{q_{n_1^1}(t)}{g_{n_1^1, n_1^2}}\right)}{g_{n_1^2, n_1^3}} 
\\& \quad + \frac{q_{n_1^3}\left(t + \frac{q_{n_1^1}(t)}{g_{n_1^1, n_1^2}} + \frac{q_{n_1^2}\left(t + \frac{q_{n_1^1}(t)}{g_{n_1^1, n_1^2}}\right)}{g_{n_1^2, n_1^3}}\right)}{\mu}
\end{aligned}
$$
which is similar for $D_2(t)$, $D_3(t)$, and $D_4(t)$. We only show the derivation associated with packets that are injected into the network at node $n_1^1$. For simplicity, we only show the derivation of the case where all the nodes have positive length during overload, where the min-delay conditions also hold for the other cases. We can simplify $D_1(t)$ as
$$
\small
\begin{aligned}
D_1(t) &= q_{n_1^1}(t) \left(\frac{g_{n_1^2,n_1^3} + g_{n_2^2, n_1^3}}{\mu} \times \frac{g_{n_1^1,n_1^2} + g_{n_2^1, n_1^2}}{g_{n_1^2,n_1^3}} \times \frac{1}{g_{n_1^1,n_1^2}} \right) 
\\& \quad + q_{n_1^2}(t)\times\frac{g_{n_1^2,n_1^3} + g_{n_2^2,n_1^3}}{\mu} \times \frac{1}{g_{n_1^2,n_1^3}} + q_{n_1^3}\frac{1}{\mu} 
\end{aligned}
$$
and by expressing $q_{i}(t)$ as the multiplication of $t$ (assuming $t_0=0$) and the queue growth rate, we have
$$
\begin{aligned}
\bar{D}_{1} &:= \frac{1}{T}\int_{0}^T D_1(t) dt
\\& = \frac{T}{2\mu} \left( \lambda_1 \times \frac{g_{n_1^2, n_1^3} + g_{n_2^2, n_1^3}}{g_{n_1^2, n_1^3}} \times \frac{g_{n_1^1, n_1^2} + g_{n_2^1, n_1^2}}{g_{n_1^1, n_1^2}} - \mu \right).
\end{aligned}
$$
We can now express the average delay $\bar{D}_{\text{avg}}$ as
\begin{equation}
\small
\label{eqn:tree_d_avg}
\begin{aligned}
&\bar{D}_{\text{avg}} := \frac{\sum_{i=1}^4 \lambda_i \bar{D}_i}{\sum_{i=1}^4 \lambda_i}
\\& 
\overset{(a)}{\sim} \frac{g_{n_1^2,n_1^3} + g_{n_2^2,n_1^3}}{g_{n_1^2,n_1^3}} \left(\lambda_1^2 \frac{g_{n_1^1,n_1^2} + g_{n_2^1,n_1^2}}{g_{n_1^1,n_1^2}} + \lambda_2^2 \frac{g_{n_1^1,n_1^2} + g_{n_2^1,n_1^2}}{g_{n_2^1,n_1^2}}\right)
\\& \quad 
+ \frac{g_{n_1^2,n_1^3} + g_{n_2^2,n_1^3}}{g_{n_2^2,n_1^3}}  
\left(\lambda_3^2 \frac{g_{n_3^1,n_2^2} + g_{n_4^1,n_2^2}}{g_{n_3^1,n_2^2}} + \lambda_4^2 \frac{g_{n_3^1,n_2^2} + g_{n_4^1,n_2^2}}{g_{n_4^1,n_2^2}}\right)
\\& \overset{(b)}{\geq}  \frac{g_{n_1^2,n_1^3} + g_{n_2^2,n_1^3}}{g_{n_1^2,n_1^3}} (\lambda_1 + \lambda_2)^2 + \frac{g_{n_1^2,n_1^3} + g_{n_2^2,n_1^3}}{g_{n_2^2,n_1^3}}  (\lambda_3 + \lambda_4)^2
\\& \overset{(c)}{\geq} (\lambda_1 + \lambda_2 + \lambda_3 + \lambda_4)^2
\end{aligned}
\end{equation}
where at (a) we remove the constant additive and multiplicative terms, and at (b) and (c) we apply the Cauchy-Schwarz inequality. The conditions to achieve equality at (b) are 
\begin{equation}
\label{eqn:tree_eqn_condition_1}
\frac{g_{n_1^1, n_1^2}}{g_{n_2^1, n_1^2}} = \frac{\lambda_1}{\lambda_2},\qquad \frac{g_{n_3^1, n_2^2}}{g_{n_4^1, n_2^2}} = \frac{\lambda_3}{\lambda_4},
\end{equation}
and the condition to achieve equality at (c) is
\begin{equation}
\label{eqn:tree_eqn_condition_2}
\frac{g_{n_1^2, n_1^3}}{g_{n_2^2, n_1^3}} = \frac{\lambda_1 + \lambda_2}{\lambda_3 + \lambda_4}.
\end{equation}
Therefore any policy that satisfies \eqref{eqn:tree_eqn_condition_1} and \eqref{eqn:tree_eqn_condition_2} leads to minimum $\bar{D}_{\text{avg}}$, which is $\frac{T}{2\mu}\max\left\{\left(\sum_{i=1}^4 \lambda_i - \mu \right), 0\right\}$. Similarly, we can show \eqref{eqn:tree_eqn_condition_1} and \eqref{eqn:tree_eqn_condition_2} together form a sufficient condition to minimize the maximum ingress delay $\bar{D}_{\text{max}}$.

\end{document}